\documentclass[a4paper,UKenglish,cleveref,thm-restate]{lipics-v2021}

\usepackage{ellipsis, mparhack, ragged2e} 
\usepackage[l2tabu, orthodox]{nag} 

\usepackage[T1]{fontenc}
\usepackage[utf8]{inputenc}

\usepackage{csquotes}
\usepackage{booktabs}

\usepackage{pgfplots}
\pgfplotsset{compat=1.14}
\usepackage{tikz}
\usetikzlibrary{shapes,positioning,arrows,arrows.meta,calc,automata,matrix,fit,backgrounds}
\tikzstyle{state}+=[minimum size = 6mm, inner sep=0,outer sep=1]
\tikzset{->,>=stealth'}

\colorlet{disabled}{lightgray}
\tikzstyle{state}=[draw,rectangle,inner sep=5pt,rounded corners=2pt]
\tikzstyle{action}=[font=\small,inner sep=0pt,outer sep=3pt]
\tikzstyle{actionnode}=[circle,draw=black,fill=black,minimum size=1mm,inner sep=0,outer sep=0]
\tikzstyle{actionedge}=[draw,-]
\tikzstyle{prob}=[font=\scriptsize,inner sep=0pt,outer sep=1pt]
\tikzstyle{probedge}=[draw,->]
\tikzstyle{directedge}=[draw,->]
\tikzset{chainarrow/.tip={Stealth[length=3pt]}}
\tikzset{>=chainarrow}

\usepackage{xparse}
\usepackage{mathtools}
\usepackage{environ}
\usepackage{marvosym}
\usepackage{bbm}

\usepackage{amssymb}
\usepackage{amsmath}
\usepackage{amsthm}
\usepackage{cleveref}
\usepackage{subcaption}

\usepackage{algorithmicx,algorithm}
\usepackage[noend]{algpseudocode}

\makeatletter
\newcounter{algorithmicH}
\let\oldalgorithmic\algorithmic
\renewcommand{\algorithmic}{%
  \stepcounter{algorithmicH}
  \oldalgorithmic}
\renewcommand{\theHALG@line}{ALG@line.\thealgorithmicH.\arabic{ALG@line}}
\makeatother

\DeclarePairedDelimiter{\delimabs}{\lvert}{\rvert}
\DeclarePairedDelimiter{\delimcardinality}{\lvert}{\rvert}
\DeclarePairedDelimiter{\delimnorm}{\lVert}{\rVert}

\NewDocumentCommand{\abs}{sm}{\IfBooleanTF{#1}{\delimabs*{#2}}{\delimabs{#2}}}
\NewDocumentCommand{\cardinality}{sm}{\IfBooleanTF{#1}{\delimcardinality*{#2}}{\delimcardinality{#2}}}
\NewDocumentCommand{\norm}{sm}{\IfBooleanTF{#1}{\delimnorm*{#2}}{\delimnorm{#2}}}
\NewDocumentCommand{\powerset}{r()}{2^{#1}}
\newcommand{\setcomplement}[1]{\overline{#1}}
\newcommand{\indicator}[1]{\mathbbm{1}_{#1}}
\newcommand{\lipschitz}{L}

\newcommand{\unionSym}{\cup}
\newcommand{\unionBin}{\mathbin{\unionSym}}

\newcommand{\intersectionSym}{\cap}
\newcommand{\intersectionBin}{\mathbin{\intersectionSym}}
\newcommand{\UnionSym}{\bigcup}

\newcommand{\union}{\unionBin}

\newcommand{\intersection}{\intersectionBin}
\newcommand{\Union}{\UnionSym}

\newcommand{\Naturals}{\mathbb{N}}
\newcommand{\Reals}{\mathbb{R}}

\DeclareMathOperator{\support}{supp}
\DeclareMathOperator*{\argmax}{arg\,max}
\DeclareMathOperator{\metric}{d}

\newcommand{\measure}{\mu}
\NewDocumentCommand{\Measures}{d()}{\IfNoValueTF{#1}{\Pi}{\Pi(#1)}}
\NewDocumentCommand{\integral}{d<> m m}{\IfNoValueTF{#1}{\int #2\,d#3}{\int_{#1} #2\,d#3}}
\NewDocumentCommand{\Expectation}{s d[]}{\IfNoValueTF{#2}{\mathbb{E}}{\mathbb{E}\IfBooleanTF{#1}{\left[#2\right]}{[#2]}}}
\NewDocumentCommand{\Probability}{s d[]}{\IfNoValueTF{#2}{\mathbb{P}}{\mathbb{P}\IfBooleanTF{#1}{\left[#2\right]}{[#2]}}}
\newcommand{\sigmaalgebra}{\Sigma}
\newcommand{\totalvariation}{\delta_{TV}}

\newcommand{\ltrue}{\mathsf{true}}
\newcommand{\lfalse}{\mathsf{false}}

\newcommand{\MDP}{\mathcal{M}}

\newcommand{\States}{S}
\newcommand{\initialstate}{s_{0}}
\newcommand{\Actions}{Act}
\NewDocumentCommand{\stateactions}{d()}{\IfNoValueTF{#1}{{Av}}{{Av}(#1)}}
\NewDocumentCommand{\mctransitions}{d()}{\IfNoValueTF{#1}{\delta}{\delta(#1)}}
\NewDocumentCommand{\mdptransitions}{d()}{\IfNoValueTF{#1}{\Delta}{\Delta(#1)}}
\NewDocumentCommand{\actionstate}{r<> r()}{s_{#1}(#2)}

\newcommand{\sigmaalgebrastates}{\sigmaalgebra_{\States}}
\newcommand{\sigmaalgebraactions}{\sigmaalgebra_{\Actions}}
\newcommand{\metricstates}{\metric_{\States}}
\newcommand{\metricactions}{\metric_{\Actions}}
\newcommand{\metricproduct}{\metric_{\times}}
\newcommand{\lipschitzstates}{C_{\States}}

\newcommand{\lipschitzproduct}{C_{\times}}
\newcommand{\sampled}{\mathsf{Sampled}}

\newcommand{\infinitepath}{\rho}
\newcommand{\finitepath}{\varrho}
\NewDocumentCommand{\Infinitepaths}{d<>}{\IfNoValueTF{#1}{\mathsf{Paths}}{\mathsf{Paths}_{#1}}}
\NewDocumentCommand{\Finitepaths}{d<>}{\IfNoValueTF{#1}{\mathsf{FPaths}}{\mathsf{FPaths}_{#1}}}
\newcommand{\strategy}{\pi}
\NewDocumentCommand{\Strategies}{d<>}{\IfNoValueTF{#1}{\Pi}{\Pi_{#1}}}
\NewDocumentCommand{\StrategiesM}{d<>}{\IfNoValueTF{#1}{\Pi}{\Pi_{#1}}^{\mathsf{M}}}
\NewDocumentCommand{\StrategiesMD}{d<>}{\IfNoValueTF{#1}{\Pi}{\Pi_{#1}}^{\mathsf{MD}}}
\newcommand{\last}[1]{last(#1)}

\DeclareMathOperator{\SccsOp}{SCC}
\DeclareMathOperator{\BsccsOp}{BSCC}
\DeclareMathOperator{\EcsOp}{EC}
\DeclareMathOperator{\MecsOp}{MEC}

\NewDocumentCommand{\Sccs}{r()}{\SccsOp(#1)}
\NewDocumentCommand{\Bsccs}{r()}{\BsccsOp(#1)}
\NewDocumentCommand{\Ecs}{d()}{\IfNoValueTF{#1}{\EcsOp}{\EcsOp(#1)}}
\NewDocumentCommand{\Mecs}{d()}{\IfNoValueTF{#1}{\MecsOp}{\MecsOp(#1)}}

\NewDocumentCommand{\ProbabilityMC}{s r<> d[]}{\mathsf{Pr}_{#2}\IfNoValueF{#3}{\IfBooleanTF{#1}{\!\left[#3\right]\!}{[#3]}}}
\NewDocumentCommand{\ProbabilityMDP}{s r<> r<> d[]}{\mathsf{Pr}_{#2}^{#3}\IfNoValueF{#4}{\IfBooleanTF{#1}{\!\left[#4\right]\!}{[#4]}}}
\NewDocumentCommand{\ProbabilityMDPmax}{s r<> d[]}{\mathsf{Pr}_{#2}^{\max}\IfNoValueF{#3}{\IfBooleanTF{#1}{\!\left[#3\right]\!}{[#3]}}}
\NewDocumentCommand{\ProbabilityMDPsup}{s r<> d[]}{\mathsf{Pr}_{#2}^{\sup}\IfNoValueF{#3}{\IfBooleanTF{#1}{\!\left[#3\right]\!}{[#3]}}}

\NewDocumentCommand{\ExpectedSum}{m m}{#1\langle#2\rangle}

\NewDocumentCommand{\ExpectedSumMDP}{m m m m}{\ExpectedSum{#1(#2,#3)}{#4}}

\newcommand{\reach}{\Diamond}

\newcommand{\val}{\mathcal{V}}
\newcommand{\upperbound}{\mathsf{U}}
\newcommand{\lowerbound}{\mathsf{L}}
\newcommand{\bounddifference}{\mathrm{Diff}}
\newcommand{\upperboundstored}{\widehat{\mathsf{U}}}
\newcommand{\lowerboundstored}{\widehat{\mathsf{L}}}

\newcommand{\targetset}{T}
\newcommand{\sinkset}{R}

\newcommand{\algostep}{\mathsf{t}}
\newcommand{\getpair}{\textsc{GetPair}}
\newcommand{\underapprox}{\textsc{Approx}_\leq}
\newcommand{\overapprox}{\textsc{Approx}_\geq}
\newcommand{\getprecision}{\textsc{Precision}(\algostep)}

\newcommand{\ltlUntil}{\mathbin{\mathbf{U}}}
\newcommand{\ltlNext}{\mathbin{\mathbf{X}}}
\newcommand{\ltlFinally}{\mathop{\mathbf{F}}}
\newcommand{\ltlGlobally}{\mathop{\mathbf{G}}}
\AtBeginDocument{%
	\providecommand\BibTeX{{%
			\normalfont B\kern-0.5em{\scshape i\kern-0.25em b}\kern-0.8em\TeX}}}

\usepackage{todonotes}

\newtoggle{arxiv}
\newcommand{\arxivref}[2]{\iftoggle{arxiv}{#1}{#2}}
\toggletrue{arxiv}
\newcommand{\appendixref}{App.~}

\title{Anytime Guarantees for Reachability in Uncountable Markov Decision Processes}

\author{Kush Grover}{Technical University of Munich, Germany}{kush.grover@in.tum.de}{https://orcid.org/0000-0003-4575-1302}{}
\author{Jan K\v{r}et{\'i}nsk{\'y}}{Technical University of Munich, Germany}{jan.kretinsky@in.tum.de}{https://orcid.org/0000-0002-8122-2881}{}
\author{Tobias Meggendorfer}{Institute of Science and Technology Austria, Vienna, Austria}{tobias.meggendorfer@ist.ac.at}{https://orcid.org/0000-0002-1712-2165}{}
\author{Maximilian Weininger}{Technical University of Munich, Germany}{maxi.weininger@tum.de}{https://orcid.org/0000-0002-0163-2152}{}

\authorrunning{K. Grover, J. K\v{r}et{\'i}nsk{\'y}, T. Meggendorfer, and M. Weininger}

\Copyright{K. Grover, J. K\v{r}et{\'i}nsk{\'y}, T. Meggendorfer, and M. Weininger}

\ccsdesc[500]{Mathematics of computing~Markov processes}
\ccsdesc[300]{Mathematics of computing~Continuous mathematics}
\ccsdesc[300]{Computing methodologies~Continuous models}

\keywords{Uncountable system, Markov decision process, discrete-time Markov control process, probabilistic verification, anytime guarantee}
\iftoggle{arxiv}{
\pdfoutput=1
\hideLIPIcs
\nolinenumbers
}{
\relatedversiondetails[cite=techreport]{Full version}{https://arxiv.org/abs/2008.04824} 
}

\EventEditors{Bartek Klin, S\l{}awomir Lasota, and Anca Muscholl}
\EventNoEds{3}
\EventLongTitle{33rd International Conference on Concurrency Theory (CONCUR 2022)}
\EventShortTitle{CONCUR 2022}
\EventAcronym{CONCUR}
\EventYear{2022}
\EventDate{September 12--16, 2022}
\EventLocation{Warsaw, Poland}
\EventLogo{}
\SeriesVolume{243}
\ArticleNo{17}

\begin{document}

\maketitle

\begin{abstract}
We consider the problem of approximating the reachability probabilities in Markov decision processes (MDP) with uncountable (continuous) state and action spaces.
While there are algorithms that, for special classes of such MDP, provide a sequence of approximations converging to the true value in the limit, our aim is to obtain an algorithm with guarantees on the precision of the approximation.

As this problem is undecidable in general, assumptions on the MDP are necessary.
Our main contribution is to identify sufficient assumptions that are as weak as possible, thus approaching the \enquote{boundary} of which systems can be correctly and reliably analyzed.
To this end, we also argue why each of our assumptions is necessary for algorithms based on processing finitely many observations.

We present two solution variants.
The first one provides converging lower bounds under weaker assumptions than typical ones from previous works concerned with guarantees.
The second one then utilizes stronger assumptions to additionally provide converging upper bounds.
Altogether, we obtain an \emph{anytime} algorithm, i.e. yielding a sequence of approximants with known and iteratively improving precision, converging to the true value in the limit.
Besides, due to the generality of our assumptions, our algorithms are very general templates, readily allowing for various heuristics from literature in contrast to, e.g., a specific discretization algorithm.
Our theoretical contribution thus paves the way for future practical improvements without sacrificing correctness guarantees.

\end{abstract}

\setlength{\abovedisplayskip}{0.5\abovedisplayskip}
\setlength{\abovedisplayshortskip}{0.5\abovedisplayshortskip}
\setlength{\belowdisplayskip}{0.5\belowdisplayskip}
\setlength{\belowdisplayshortskip}{0.5\belowdisplayshortskip}

\newpage

\section{Introduction} \label{sec:intro}
%
The standard formalism for modelling systems with both non-deterministic and probabilistic behaviour are Markov decision processes (MDP) \cite{DBLP:books/wi/Puterman94}. 
In the context of many applications such as cyber-physical systems, states and actions are used to model real-valued phenomena like position or throttle.
Consequently, the state space and the action space may be uncountably infinite.
For example, the intervals $[a, b] \times [c, d] \subseteq \Reals^2$ can model a safe area for a robot to move in or a set of available control inputs such as acceleration and steering angle.
This gives rise to MDP with uncountable state- and action-spaces (sometimes called controlled discrete-time Markov process \cite{DBLP:conf/hybrid/TkachevMKA13,DBLP:journals/iandc/TkachevMKA17} or discrete-time Markov control process \cite{chatterjee2011maximizing,hernandez2012discrete}), with applications ranging from modelling a Mars rover \cite{DBLP:journals/corr/abs-1301-0559,DBLP:journals/corr/abs-1902-00778}, over water reservoir control \cite{lamond2002water} and warehouse storage management \cite{mahootchi2009storage}, to energy control \cite{DBLP:conf/hybrid/TkachevMKA13}, and many more~\cite{peskir2006optimal}.

Although systems modelled by MDP are often safety-critical, the analysis of uncountable systems is so complex that practical approaches for verification and controller synthesis are usually based on \enquote{best effort} learning techniques, for example \emph{reinforcement learning}.
While efficient in practice, these methods guarantee, even in the best case, convergence to the true result only in the limit, e.g.\ \cite{DBLP:conf/icml/MeloMR08}, or for increasingly precise discretization, e.g.\ \cite{DBLP:conf/hybrid/TkachevMKA13,DBLP:conf/atva/JaegerJLLST19}.
In line with the tradition of learning and to make the analysis more feasible, the typical objectives considered for MDP are either finite-horizon \cite{DBLP:conf/aaai/LiL05,DBLP:journals/automatica/AbatePLS08} or discounted properties \cite{DBLP:conf/uai/GuestrinHK04,DBLP:books/sp/12/Hasselt12,DBLP:journals/tac/HaskellJSY20}, together with restrictive assumptions.
Note that when it comes to approximation, discounted properties effectively are finite-horizon.
In contrast, ensuring safety of a reactive system or a certain probability to satisfy its mission goals requires an \emph{unbounded} horizon and reduces to optimizing the reachability probabilities.
Moreover, the safety-critical context requires \emph{reliable} bounds on the probability, not an approximation with \emph{unknown} precision.

In this paper, we provide the first provably correct anytime algorithm for (unbounded) reachability in uncountable MDP.
As an \emph{anytime} algorithm, it can at every step of the execution return correct lower and upper bounds on the true value.
Moreover, these bounds gradually converge to the true value, allowing approximation up to an arbitrary precision.
Since the problem is undecidable, the core of our contribution is identifying sufficient conditions on the uncountable MDP to allow for approximation.

Our \emph{primary goal} is to provide \emph{conditions as weak as possible}, thereby pushing towards the boundary of which systems can be analyzed provably correctly.
To this end, we do not rely on any particular representation of the system.
Nonetheless, for classical scenarios, and, in particular, for finite MDP, our conditions are mostly satisfied trivially.

Our \emph{secondary goal} is to derive the respective algorithms as an \emph{extension of value iteration} (VI) \cite{howard1960dynamic,DBLP:books/wi/Puterman94}, while \emph{avoiding drawbacks of discretization}-based approaches.
VI is a de facto standard method for numerical analysis of finite MDP, in particular with reachability objectives, regarded as practically efficient and allowing for heuristics avoiding the exploration of the complete state space, e.g. \cite{DBLP:conf/atva/BrazdilCCFKKPU14}.
Interestingly, even for finite MDP, anytime VI algorithms with precision guarantees are quite recent \cite{DBLP:conf/atva/BrazdilCCFKKPU14,DBLP:journals/tcs/HaddadM18,DBLP:conf/cav/Baier0L0W17,DBLP:conf/cav/QuatmannK18,DBLP:conf/cav/HartmannsK20}.
Previous to that, the most used model checkers could return arbitrarily wrong results \cite{DBLP:journals/tcs/HaddadM18}.
Providing VI with precision guarantees for general uncountable MDP is thus worthwhile on its own.
Finally, while discretization is conceptually simple, we prefer to provide a solution that avoids the need to introduce arbitrary boundaries through gridding the whole state space and, moreover, instead utilizes information from one \enquote{cell} of the grid in other places, too.

To summarize, while algorithmic aspects form an important motivation, our primary contribution is theoretical: an explicit and complete set of generic assumptions allowing for guarantees, disregarding practical efficiency at this point.
Consequently, while our approach lays foundations for further, more tailored approaches, it is not to be seen as a competitor to the existing practical, best-effort techniques, as these aim for a completely different goal.

\noindent
\textbf{Our Contribution}
In this work, we provide the following:
\begin{description}
	\item[\cref{sec:cvi}:]
	A set of assumptions that allow for computing converging \emph{lower} bounds on the reachability probability in MDP with uncountable state and action spaces.
	We discuss in detail why they are weaker than usual, necessary, and applicable to typically considered systems.
	With these assumptions, we extend the standard (convergent but precision-ignorant) VI to this general setting.

	\item[\cref{sec:cbrtdp}:]
	An additional set of assumptions that yield the first \emph{anytime} algorithm, i.e.\ with provable bounds on the precision/error of the result, converging to 0.
	We combine the preceding algorithm with the technique of \emph{bounded real-time dynamic programming (BRTDP)} \cite{DBLP:conf/icml/McMahanLG05} and provide also converging \emph{upper} bounds on the reachability probability.

	\item[\cref{sec:discussion}:]
	A discussion of theoretical extensions and practical applications. 
\end{description}

\noindent
\textbf{Related work}
For detailed \emph{theoretical} treatment of reachability and related problems on uncountable MDP, see e.g.\ \cite{DBLP:journals/iandc/TkachevMKA17,chatterjee2011maximizing}.
Reachability on uncountable MDP generalizes numerous problems known to be undecidable.
For example, we can encode the halting problem of (probabilistic) Turing machines by encoding the tape content as real value.
Similarly, almost-sure termination of probabilistic programs (undecidable \cite{DBLP:conf/mfcs/KaminskiK15}) is a special case of reachability on general uncountable MDP (see e.g.\ \cite{DBLP:conf/vmcai/FuC19}).
As precise reachability analysis is undecidable even for non-stochastic linear hybrid systems \cite{DBLP:conf/stoc/HenzingerKPV95}, many works turn their attention to more relaxed notions such as $\delta$-reachability, e.g.\ \cite{DBLP:conf/hybrid/ShmarovZ15}, and/or employ many assumptions. 

In order to obtain precision bounds, we assume that the \emph{value} function, mapping states to their reachability probability, is \emph{Lipschitz continuous} (and that we know the Lipschitz constant).
This is slightly \emph{weaker} than the classical approach of assuming Lipschitz continuity of the \emph{transition} function (and knowledge of the constant), e.g.\ \cite{DBLP:journals/ejcon/AbateKLP10,DBLP:conf/qest/SoudjaniA11}.
In particular, these assumptions (i)~imply our assumption (as we show in \arxivref{\cref{app:assumptions:lipschitz:implication}}{\cite[\appendixref{}B.2.1]{techreport}}) and (ii)~are used even in the simpler settings of finite-horizon and discounted reward scenarios \cite{bertsekas1975convergence,DBLP:journals/ejcon/AbateKLP10,DBLP:conf/qest/SoudjaniA11,DBLP:conf/hybrid/TkachevMKA13} or even more restricted settings to obtain practical efficiency, e.g.\ \cite{DBLP:conf/qest/LalP18}.
In contrast to our approach, they are not anytime algorithms and require treatment of the whole state space.

To provide context, we outline how continuity is used (explicitly or implicitly) in related work and mention their respective results.
	Firstly,~\cite{DBLP:journals/tac/HaskellJSY20,DBLP:conf/uai/SharmaJ019} assume Lipschitz continuity, \emph{but not} explicit knowledge of the constant.
	In essence, these approaches solve the problem by successively increasing internal parameters.
	The parameters then eventually cross a bound implied by the Lipschitz constant, yielding an \enquote{eventual correctness}.
	In particular, they provide \enquote{convergence in the limit} or \enquote{probably approximately correct} results, but no bounds on the error or the convergence rate; these would depend on knowledge of the constant.
	
	Secondly, \cite{DBLP:conf/uai/GuestrinHK04,DBLP:conf/icml/MeloMR08,DBLP:journals/ejcon/AbateKLP10,DBLP:conf/qest/SoudjaniA11,DBLP:conf/hybrid/TkachevMKA13} (and our work) assume Lipschitz continuity \emph{and} knowledge of the constant.
	Relying on the constant being provided externally, these works derive guarantees.
	Previously, the guarantees given are weaker than our convergent anytime bounds:
	Either convergence in the limit~\cite{DBLP:conf/icml/MeloMR08} or a bound on a discretization error, relativized to sub-optimal strategies~\cite{DBLP:conf/uai/GuestrinHK04} or bounded horizon~\cite{DBLP:journals/ejcon/AbateKLP10,DBLP:conf/qest/SoudjaniA11,DBLP:conf/hybrid/TkachevMKA13}.

	
%
Several of the above mentioned works employ \emph{discretization}~\cite{DBLP:conf/uai/GuestrinHK04,DBLP:journals/ejcon/AbateKLP10,DBLP:conf/qest/SoudjaniA11,DBLP:conf/hybrid/TkachevMKA13}. This method is quite general, but obtaining any bounds on the error requires continuity assumptions~\cite{DBLP:conf/hybrid/AbateAPLS07}.
Further, there are works that use other assumptions:
\cite{DBLP:conf/atal/HasanbeigAK19,DBLP:journals/corr/abs-1902-00778} use \emph{reinforcement learning} methods to tackle reachability and more general problems, without any continuity assumption.
However, they do not provide any guarantees.
See \cite{DBLP:books/sp/12/Hasselt12} for a detailed exposition of similar approaches.
Assuming an abstraction is given, \emph{abstraction and bisimulation} approaches, e.g.\ \cite{DBLP:journals/siamco/HaesaertSA17,DBLP:conf/adhs/HaesaertSA18}, provide guarantees, but only on the lower bounds.
With significant assumptions on the system's structure, \emph{symbolic} approaches \cite{DBLP:conf/aaai/LiL05,DBLP:conf/bracis/ViannaSB14,DBLP:conf/uai/SannerDB11,DBLP:conf/uai/FengDMW04} may even obtain exact solutions.

\section{Preliminaries} \label{sec:prelim}

In this section, we recall basics of probabilistic systems and set up the notation.
As usual, $\Naturals$ and $\Reals$ refer to the (positive) natural numbers and real numbers, respectively.
For a set $S$, $\indicator{S}$ denotes its \emph{characteristic function}, i.e.\ $\indicator{S}(x) = 1$ if $x \in S$ and $0$ otherwise.
We write $S^\star$ and $S^\omega$ to refer to the set of finite and infinite sequences comprising elements of $S$, respectively.

We assume familiarity with basic notions of measure theory, e.g.\ \emph{measurable set} or \emph{measurable function}, as well as probability theory, e.g.\ \emph{probability spaces} and \emph{measures} \cite{billingsley2012probability}.
For a measure space $X$ with sigma-algebra $\sigmaalgebra_X$, $\Measures(X)$ denotes the set of all probability measures on $X$.
For a measure $\mu \in \Measures(X)$, we write $\mu(Y) = \integral{\indicator{Y}}{\mu}$ to denote the mass of a measurable set $Y \in \sigmaalgebra_X$ (also called \emph{event}).
For two probability measures $\mu$ and $\nu$, the \emph{total variation distance} is defined as $\totalvariation(\mu, \nu) := 2 \cdot \sup_{Y \in \Sigma_X} \abs{\mu(Y) - \nu(Y)}$.
Some event happens \emph{almost surely} (a.s.) w.r.t.\ some measure $\mu$ if it happens with probability $1$.
We write $\support(\mu)$ to denote the \emph{support} of the probability measure $\mu$.

\begin{remark}
	It is surprisingly difficult to give a well-defined notion of support for measures in general.
	Intuitively, $\support(\mu)$ describes the \enquote{smallest} set which $\mu$ assigns a value of $1$.
	However, this is not well-defined for general measures.
	We discuss these issues and a proper definition in \arxivref{\cref{app:proofs}}{\cite[\appendixref{}E]{techreport}}.
	Throughout this work, similar subtle issues related to measure theory arise.
	For the sake of readability, these are mostly delegated to footnotes or the \arxivref{appendix}{appendix of the full version~\cite{techreport}}, and readers may safely skip over these points.
\end{remark}
%
%
We work with \emph{Markov decision processes} (MDP) \cite{DBLP:books/wi/Puterman94}, a widely used model to capture both non-determinism and probability. We consider uncountable state and action spaces.
%
\begin{definition}
	A \emph{(continuous-space, discrete-time) Markov decision process (MDP)} is a tuple $\MDP = (\States, \Actions, \stateactions, \mdptransitions)$, where $\States$ is a compact set of \emph{states} (with topology $\mathcal{T}_\States$ and Borel $\sigma$-algebra $\sigmaalgebrastates = \mathfrak{B}(\mathcal{T}_\States)$), $\Actions$ is a compact set of \emph{actions} (with topology $\mathcal{T}_{\Actions}$ and Borel $\sigma$-algebra $\sigmaalgebraactions = \mathfrak{B}(\mathcal{T}_{\Actions})$), $\stateactions\colon \States \to \sigmaalgebraactions \setminus \{\emptyset\}$ assigns to every state a non-empty, measurable, and compact set of \emph{available actions}, and $\mdptransitions\colon \States \times \Actions \to \Measures(\States)$ is a \emph{transition function} that for each state $s$ and (available) action $a \in \stateactions(s)$ yields a probability measure over successor states (i.e.\ a Markov Kernel).
	An MDP is called \emph{finite} if $\cardinality{\States} < \infty$ and $\cardinality{\Actions} < \infty$.
\end{definition}
See \cite[Sec.~2.3]{DBLP:books/wi/Puterman94} and \cite[Chp.~9]{bertsekas1978stochastic} for a more detailed discussion on the technical considerations arising from uncountable state and action spaces.
Note that we assume the set of available actions to be non-empty.
This means that the system can never get \enquote{stuck} in a degenerate state without successors.
\emph{Markov chains} are a special case of MDP where $\cardinality{\stateactions(s)} = 1$ for all $s \in \States$, i.e.\ a completely probabilistic system without any non-determinism.
Our presented methods thus are directly applicable to Markov chains as well.


Given a measure $\measure \in \Measures(X)$ and a measurable function $f \colon X \to \Reals$ mapping elements of a set $X$ to real numbers, we write $\ExpectedSum{\measure}{f} := \integral{f(x)}{\measure(x)}$ to denote the integral of $f$ with respect to $\measure$.
For example, $\ExpectedSumMDP{\mdptransitions}{s}{a}{f}$ denotes the expected value $\mathbb E_{s'\sim\Delta(s,a)}f(s')$ of $f \colon \States \to \Reals$ over the successors of $s$ under action $a$.
Moreover, abusing notation, for some set of state $S' \subseteq \States$ and function $\stateactions' \colon S' \to \Actions$, we write $S' \times \stateactions' = \{(s, a) \mid s \in S', a \in \stateactions'(s)\}$ to denote the set of state-action pairs with states from $S'$ under $\stateactions'$.

%

An \emph{infinite path} in an MDP is some infinite sequence $\infinitepath = s_1 a_1 s_2 a_2 \cdots \in (\States \times \stateactions)^\omega$, such that for every $i \in \Naturals$ we have $s_{i+1} \in \support(\mdptransitions(s_i, a_i))$.
A \emph{finite path} (or \emph{history}) $\finitepath = s_1 a_1 s_2 a_2 \dots s_n \in (\States \times \stateactions)^\star \times \States$ is a non-empty, finite prefix of an infinite path of length $\cardinality{\finitepath} = n$, ending in state $s_n$, denoted by $\last{\finitepath}$.
We use $\infinitepath(i)$ and $\finitepath(i)$ to refer to the $i$-th state in an (in)finite path.
We refer to the set of finite (infinite) paths of an MDP $\MDP$ by $\Finitepaths<\MDP>$ ($\Infinitepaths<\MDP>$).
Analogously, we write $\Finitepaths<\MDP, s>$ ($\Infinitepaths<\MDP, s>$) for all (in)finite paths starting in $s$.

In order to obtain a probability measure, we first need to eliminate the non-determinism.
This is done by a so-called \emph{strategy} (also called \emph{policy}, \emph{controller}, or \emph{scheduler}).
%
	A strategy on an MDP $\MDP = (\States, \Actions, \stateactions, \mdptransitions)$ is a function $\strategy\colon \Finitepaths<\MDP> \to \Measures(\Actions)$, s.t.\ $\support(\strategy(\finitepath)) \subseteq \stateactions(\last{\finitepath})$.
	The set of all strategies is denoted by $\Strategies<\MDP>$.
%
Intuitively, a strategy is a \enquote{recipe} describing which step to take in the current state, given the evolution of the system so far.
%
%

Given an MDP $\MDP$, a strategy $\strategy \in \Strategies<\MDP>$, and an initial state $\initialstate$, we obtain a measure on the set of infinite paths $\Infinitepaths<\MDP>$, which we denote as $\ProbabilityMDP<\MDP, \initialstate><\strategy>$. 
See \cite[Sec.~2]{DBLP:books/wi/Puterman94} 
for further details.
Thus, given a measurable set $A \subseteq \Infinitepaths<\MDP>$, we can define its maximal probability starting from state $\initialstate$ under any strategy by
$
	\ProbabilityMDPsup<\MDP, \initialstate>[A] := {\sup}_{\strategy \in \Strategies<\MDP>} \ProbabilityMDP<\MDP, \initialstate><\strategy>[A].
$
Depending on the structure of $A$ it may be the case that no optimal strategy exists and we have to resort to the supremum instead of the maximum.
This may already arise for finite MDP, see~\cite{DBLP:journals/lmcs/ChatterjeeKK17}.

%
For an MDP $\MDP = (\States, \Actions, \stateactions, \mdptransitions)$ and a set of \emph{target states} $\targetset \subseteq \States$, \emph{(unbounded) reachability} refers to the set $\reach \targetset = \{\infinitepath \in \Infinitepaths<\MDP> \mid \exists i \in \Naturals.~\infinitepath(i) \in \targetset\}$, i.e.\ all paths which eventually reach $\targetset$.
The set $\reach T$ is measurable if $\targetset$ is measurable \cite[Sec.~3.1]{DBLP:conf/hybrid/TkachevMKA13}, \cite[Sec.~2]{DBLP:journals/iandc/TkachevMKA17}.

Now, it is straightforward to define the \emph{maximal reachability problem} of a given set of states.
Given an MDP $\MDP$, target set $\targetset$, and state $\initialstate$, we are interested in computing the maximal probability of eventually reaching $\targetset$, starting in state $\initialstate$.
Formally, we want to compute the \emph{value} of the state $\initialstate$, defined as
$
	\val(\initialstate) := \ProbabilityMDPsup<\MDP, \initialstate>[\reach \targetset] = {\sup}_{\strategy \in \Strategies<\MDP>} \ProbabilityMDP<\MDP, \initialstate><\strategy>[\reach \targetset].
$
This state value function satisfies a straightforward fixed point equation, namely
\begin{equation} \label{eq:value_fixpoint}
	\val(s) = 1 \quad \text{if $s \in \targetset$} \qquad \val(s) = {\sup}_{a \in \stateactions(s)} \ExpectedSumMDP{\mdptransitions}{s}{a}{\val} \qquad \text{otherwise.}
\end{equation}
Moreover, $\val$ is the \emph{smallest} fixed point of this equation \cite[Prop.~9.8, 9.10]{bertsekas1978stochastic}, \cite[Thm.~3]{DBLP:journals/iandc/TkachevMKA17}.
In our approach, we also deal with values of state-action pairs $(s, a) \in \States \times \stateactions$, where $\val(s, a) := \ExpectedSumMDP{\mdptransitions}{s}{a}{\val}$.
Intuitively, this represents the value achieved by choosing action $a$ in state $s$ and then moving optimally.
Clearly, we have that $\val(s) = \sup_{a \in \stateactions(s)} \val(s, a)$.
See \cite[Sec.~4]{DBLP:conf/sfm/ForejtKNP11} for a discussion of reachability on finite MDP and \cite{DBLP:journals/iandc/TkachevMKA17} for the general case.

In this work, we are interested in \emph{approximate solutions} due to the following two reasons.
Firstly, obtaining precise solutions for MDP is difficult already under strict assumptions and undecidable in our general setting.\footnote{For example, one can encode the tape of a Turing machine into the binary representation of a real number and reduce the halting problem to a reachability query.}
We thus resort to approximation, allowing for much lighter assumptions.
Secondly, by considering approximation we are able to apply many different optimization techniques, potentially leading to algorithms which are able to handle real-world systems, which are out of reach for precise algorithms even for finite MDP \cite{DBLP:conf/atva/BrazdilCCFKKPU14}.

We are interested in two types of approximations.
Firstly, we consider approximating the value function in the limit, without knowledge about how close we are to the true value.
This is captured by a semi-decision procedure for queries of the form $\ProbabilityMDPsup<\MDP, s>[\reach \targetset] > \xi$ for a threshold $\xi \in [0, 1]$.
We call this problem \textsf{ApproxLower}.
Secondly, we consider the variant where we are given a precision requirement $\varepsilon > 0$ and obtain $\varepsilon$-optimal values $(l, u)$, i.e.\ values with $\val(\initialstate) \in [l, u]$ and $0 \leq u - l < \varepsilon$.
We refer to this variant as \textsf{ApproxBounds}.

\section{Converging Lower Bounds} \label{sec:cvi}

In this section, we present the first set of assumptions, enabling us to compute \emph{converging lower bounds} on the true value, solving the \textsf{ApproxLower} problem.
In \cref{sec:vi-assms}, we discuss each assumption in detail and argue on an intuitive level why it is necessary by means of counterexamples.
With the assumptions in place, in \cref{sec:cvi:algorithm} we then present our first algorithm, also introducing several ideas we employ again in the following section.

Our assumptions and algorithms are motivated by \emph{value iteration} (VI) \cite{howard1960dynamic}, which we briefly outline.
In a nutshell, VI boils down to repeatedly applying an iteration operator to a value vector $v_n$.
For example, the canonical value iteration for reachability on finite MDP starts with $v_0(s) = 1$ for all $s \in \targetset$ and $0$ otherwise and then iterates
\begin{equation}
	v_{n+1}(s) = {\max}_{a \in \Actions(s)} {\sum}_{s' \in \States} \mdptransitions(s, a, s') \cdot v_n(s') \label{eq:vi}
\end{equation}
for all $s \notin \targetset$.
The vector $v_n$ converges monotonically from below to the true value for all states.
We mention two important points.
Firstly, the iteration can be applied \enquote{asynchronously}.
Instead of updating all states in every iteration, we can pick a single state and only update its value.
The values $v_n$ still converge to the correct value as long as all states are updated infinitely often.
Secondly, instead of storing a value per state, we can store a value for each state-action pair and obtain the state value as the maximum of these values.
Both points are a technical detail for finite MDP, however they play an essential role in our uncountable variant.
See \arxivref{\cref{app:insights:finite_vi}}{\cite[\appendixref{}A.1]{techreport}} for more details on VI for finite MDP.

In the uncountable variant of \cref{eq:vi}, $v$ is a function, $\Actions(s)$ is potentially uncountable, and the sum is replaced by integration.
As in this setting the problem is undecidable, naturally we have to employ some assumptions.
Our goal is to sufficiently imitate the essence of \cref{eq:vi}, obtaining convergence without being overly restrictive.
In particular, we want to (i)~represent (an approximation of) $v_n$ using finite memory, (ii)~safely approximate the maximum and integration, and (iii)~select appropriate points to update $v_n$.

\subsection{Assumptions}\label{sec:vi-assms}

Before discussing each assumption in detail, we first put them into context.
As we argue in the following, most of our assumptions typically hold implicitly.
Still, by stating even basic computability assumptions in a form as weak as possible, we avoid \enquote{hidden} assumptions, e.g.\ by assuming that the state space is a subset of $\Reals^d$.
Two of our assumptions are more restrictive, namely \textbf{Assumption C}: \textbf{Value Lipschitz Continuity} (\cref{sec:vi-assms-lipschitz}) and, introduced later, \textbf{Assumption D}: \textbf{Absorption} (\cref{sec:brtdp-assm-absorp}).
However, they are also often used in related works, as we detail in the respective sections.
Moreover, in light of previous results, the necessity of restrictive assumptions is to be expected:
Computing bounds is hard or even undecidable already for very restricted classes.
Aside from the discussion in the introduction, we additionally mention two further cases.
In the setting of probabilistic programs (which are a very special case of uncountable MDP), deciding almost sure termination for a fixed initial state (which is a severely restricted subclass of reachability on uncountable MDP without non-determinism) is an actively researched topic with recent advances, see e.g.~\cite{DBLP:conf/aplas/HuangFC18,DBLP:journals/pacmpl/Huang0CG19}, and shown to be $\Pi_2^0$-complete~\cite{DBLP:conf/mfcs/KaminskiK15}, i.e.\ highly undecidable.
In~\cite{DBLP:journals/jcss/HenzingerKPV98} and the references therein, the authors present (un-)decidability results for hybrid automata, which are a special case of uncountable MDP \emph{without any stochastic dynamics} (flow transitions can be modelled as actions indicating the delay).
As such, it is to be expected that the general class of models we consider has to be pruned very strictly in order to hope for any decidability results.

\begin{remark}
	As already mentioned, we want to provide assumptions which are as general as possible.
	Importantly, we avoid (unnecessarily) assuming any particular representation of the system.
	Our motivation is to ultimately identify the boundary of what is necessary to derive guarantees.
	While our assumptions are motivated by VI and built around \cref{eq:vi}, we note that being able to represent the state values and evaluate (some aspect of) the transition dynamics intuitively are a necessity for \emph{any} method dealing with such systems.
	We do not claim that our framework of assumptions is the only way to approach the problem, instead we provide arguments why it is a sensible way to do so.
\end{remark}

\subsubsection{A: Basic Assumptions (Asm.\ \textbf{A1}-\textbf{A4})}\label{sec:vi-assms-basic}
We first present a set of basic computability assumptions (\textbf{A1}-\textbf{A4}).
These are essential, since for uncountable systems even the simplest computations are intractable without any assumptions.
More specifically, such systems cannot be given explicitly (due to their infinite size), but instead have to be described symbolically by, e.g., differential equations.
Thus, we necessarily require some notion of computability and structural properties for each part of this symbolic description.
And indeed, each assumption essentially corresponds to one part of the MDP description (\textbf{Metric Space} to $\States \times \Actions$, \textbf{Maximum Approximation} to $\stateactions$, \textbf{Transition Approximation} to $\mdptransitions$, and \textbf{Target Computability} to $\targetset$).
They are weak and hold on practically all commonly considered systems (see \arxivref{\cref{app:assumptions:basic}}{\cite[\appendixref{}B.1]{techreport}}).
In particular, finite MDP and discrete components are trivially subsumed by considering the discrete metric.
\begin{description}
	\item[A1: Metric Space]
		$\States$ and $\Actions$ are metric spaces with (computable) metrics $\metricstates$ and $\metricactions$, respectively, and $\metricproduct$ is a compatible\footnote{For two pairs $(s, a)$ and $(s, a')$ we have that $k \cdot \metricactions(a, a') \leq \metricproduct((s, a), (s, a')) \leq K \cdot \metricactions(a, a')$ for some constants $k, K \geq 0$, analogous for $\metricstates$, achieved by, e.g.\ $\metricproduct((s,a), (s', a')) \coloneqq \metricstates(s, s') + \metricactions(a, a')$.} metric on the space of state-action pairs $\States \times \stateactions$,

	\item[A2: Maximum Approximation]
		For each state $s$ and computable Lipschitz $f : \stateactions(s) \to [0, 1]$, the value $\max_{a \in \stateactions(s)} f(a)$ can be under-approximated to arbitrary precision.

	\item[A3: Transition Approximation]
		For each state-action pair $(s,a)$ and Lipschitz $g : \States \to [0, 1]$ which can be under-approximated to arbitrary precision, the successor expectation $\ExpectedSumMDP{\mdptransitions}{s}{a}{g}$ can be under-approximated to arbitrary precision.

	\item[A4: Target Computability]
		The target set $\targetset$ is decidable, i.e.\ we are given a computable predicate which, given a state $s$, decides whether $s \in \targetset$.
\end{description}
We denote the approximations for \textbf{A2} and \textbf{A3} by $\underapprox$, i.e.\ given a pair $(s, a)$ and functions $f$, $g$ as in the assumptions, we write (abusing notation) $\underapprox(\max_{a \in \stateactions(s)} f(a), \varepsilon)$ and $\underapprox(\ExpectedSumMDP{\mdptransitions}{s}{a}{g}, \varepsilon)$ for approximation of the respective values up to precision $\varepsilon$, i.e.\ $0 \leq \max_{a \in \stateactions(s)} f(a) - \underapprox(\max_{a \in \stateactions(s)} f(a), \varepsilon) \leq \varepsilon$ and analogous for $\ExpectedSumMDP{\mdptransitions}{s}{a}{g}$.
Note that \textbf{A2} and \textbf{A3} are satisfied if we can sample densely in $\stateactions(s)$ and approximate $\mdptransitions(s, a)$.
\subsubsection{B: Sampling (Asm.\ B.VI)}\label{sec:vi-assms-sampling}
As there are uncountably many states, we are unable to explicitly update all of them at once and instead update values asynchronously.
Moreover, as there may also be uncountably many actions, we instead store and update the values of state-action pairs.
Together, we need to pick state-action pairs to update.
We delegate this choice to a selection mechanism $\getpair$, an oracle for state-action pairs.
We allow for $\getpair$ to be \enquote{stateful}, i.e.\ the sampled state-action pair may depend on previously returned pairs.
This is required in, for example, round-robin or simulation-based approaches.
We only require a basic notion of fairness in order to guarantee that we do not miss out on any information.
Note the additional identifier \textbf{.VI} (\emph{value iteration}) on the assumption name; later on, a similar, but weaker variant (\textbf{B.BRTDP}) is introduced.
\begin{description}
	\item[B.VI: State-Action Sampling]
		Let $\States^{\reach} = \{\last{\finitepath} \mid \finitepath \in \Finitepaths<\MDP, s>\}$ the set of all reachable states.
		Then, for any $\varepsilon > 0$, $s \in \States^{\reach}$, and $a \in \stateactions(s)$ we have that $\getpair$ eventually yields a pair $(s', a')$ with $\metricproduct((s, a), (s', a')) < \varepsilon$ and $\totalvariation(\mdptransitions(s, a), \mdptransitions(s', a')) < \varepsilon$ a.s.\footnote{Technically, it is sufficient to satisfy this property on any subset of $\States^{\reach}$ which only differs from it up to measure 0.
		More precisely, we only require that this assumption holds for $\States^{\reach} = \support(\ProbabilityMDPsup<\MDP, s>)$, i.e.\ the set of all reachable paths with non-zero measure.
		We omit this rather technical notion and the discussion it entails in order to avoid distracting from the central results of this work.
		}
\end{description}
Essentially, this means that $\getpair$ provides a way to \enquote{exhaustively} generate all behaviours of the system up to a precision of $\varepsilon$.
This fairness assumption is easily satisfied under usual conditions.
For example, if $\States \times \stateactions$ is a bounded subset of $\Reals^d$, we can randomly sample points in that space or consider increasingly dense grids.
Alternatively, if we can sample from the set of actions and from the distributions of $\mdptransitions$, $\getpair$ can be implemented by sampling paths of random length, following random actions.
Note that we can view the procedure as a \enquote{template}:
Instead of requiring a concrete method to acquire pairs to update, we leave this open for generality; we discuss implications of this in \cref{sec:discussion-related,sec:discussion-practical}.

The requirement on total variation may seem unnecessary, especially given that we will also assume continuity.
However, otherwise we could, for example, miss out on solitary actions which are the \enquote{witnesses} for a state's value:
suppose that $\stateactions(s) = [0, 1]$ and $\mdptransitions(s, 0)$ moves to the goal, while $\mdptransitions(s, a)$ just loops back to $s$.
Only selecting actions close to $a = 0$ w.r.t.\ the product metric is not sufficient to observe that we can move to the goal.
Note that this would not be necessary if we assumed continuity of the transition function -- selecting \enquote{nearby} actions then also yields \enquote{similar} behaviour.


\subsubsection{C: Lipschitz Continuity}\label{sec:vi-assms-lipschitz}
Finally, we present our already advertised continuity assumption.
For simplicity, we give it in its strict form and discuss relaxations later in \cref{sec:discussion-theory}.
Intuitively, Lipschitz continuity allows us to extrapolate the behaviour of the system from a single state to its surroundings.
%
%
\begin{description}
	\item[C: Value Lipschitz Continuity]
		The value functions $\val(s)$ and $\val(s, a)$ are Lipschitz continuous with \emph{known} constants $\lipschitzstates$ and $\lipschitzproduct$, i.e.\ for all $s, s' \in \States$ and $a \in \stateactions(s), a' \in \stateactions(s')$ we have
		\begin{align*}
			\abs{\val(s) - \val(s')} &\leq \lipschitzstates \cdot \metricstates(s, s') &
			\abs{\val(s, a) - \val(s', a')} &\leq \lipschitzproduct \cdot \metricproduct((s, a), (s', a'))
		\end{align*}
\end{description}
This requirement may seem quite restrictive at first glance.
Indeed, it is the only one in this section to not usually hold on \enquote{standard} systems.
However, in order to obtain any kind of (provably correct) bounds, some notion of continuity is elementary, since otherwise we cannot safely extrapolate from finitely many observations to an uncountable set.
The immediately arising questions are (i)~why \textit{Lipschitz} continuity is necessary compared to, e.g., regular or uniform continuity, and (ii)~why \textit{knowledge of the Lipschitz constant} is required.
For the first point, note that we want to be able to extrapolate from values assigned to a single state to its immediate surroundings.
While continuity means that the values in the surroundings do not \enquote{jump}, it does not give us any way of bounding the rate of change, and this rate may grow arbitrarily (for example, consider the continuous but not Lipschitz function $\sin(\frac{1}{x})$ for $x > 0$).
So, also relating to the second point, without knowledge of the Lipschitz constant, regular continuity and Lipschitz continuity are (mostly) equivalent from a computational perspective:
The function does not have discontinuities, but we cannot safely estimate the rate of change in general.
To illustrate this point further, we give an intuitive example.
\begin{example}\label{example:lipschitz}
	We construct an MDP with a periodic, Lipschitz continuous value function, as illustrated in \cref{fig:example_lipschitz knowledge} and formally defined below.
	Intuitively, for a given period width $w$ (e.g.\ 0.25) and a periodic function $f$ (e.g.\ a triangle function), a state $s$ between $0$ and $w$ moves to a target or sink with probability $f(s)$.
	All larger states $s \geq w$ transition to $s - w$ with probability 1.
	The value function thus is periodic and Lipschitz continuous, see \cref{fig:example_lipschitz knowledge} for a possible value function and \arxivref{\cref{app:insights:lipschitz}}{\cite[\appendixref{}B.2.3]{techreport}} for a formal definition.

	\begin{figure}[t]
		\centering
		\begin{tikzpicture}
			\begin{axis}[xmin=0,xmax=1,ymin=0,ymax=1,samples=50,width=10cm,height=2.5cm,
				axis x line=middle,
				axis line style={-},
				enlarge x limits=0,
				enlarge y limits=0,
				xlabel=$\States$,ylabel=$\val(s)$,
				xtick={0,0.25,0.5,0.75,1},xticklabels={0,0.25,0.5,0.75,1},
				x label style={anchor=north},
				ytick={0,0.5,1},
				every tick/.style={-},
				]
				\addplot[-,no marks,samples=50] plot coordinates {
					(0,0)
					(0.125,1)
					(0.25,0)
					(0.375,1)
					(0.5,0)
					(0.625,1)
					(0.75,0)
					(0.875,1)
					(1,0)
				};
				\node[inner sep=0pt, outer sep=2pt, draw, black, circle] (s0) at (axis cs:0.8125,0.5) {};
				\node[inner sep=0pt, outer sep=2pt, draw, black, circle] (s1) at (axis cs:0.5625,0.5) {};
				\node[inner sep=0pt, outer sep=2pt, draw, black, circle] (s2) at (axis cs:0.3125,0.5) {};
				
			\end{axis}
			\path[->]
			(s0) edge[bend right] (s1)
			(s1) edge[bend right] (s2)
			;
		\end{tikzpicture}
		\caption{
			The value function of \cref{example:lipschitz}, showing that knowledge of the constant is important.
		} \label{fig:example_lipschitz knowledge}
	\end{figure}
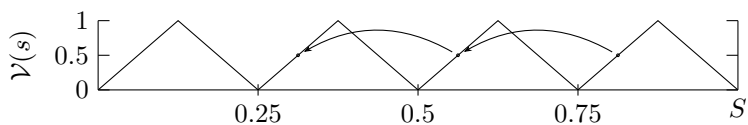

	For a finite number of samples, we can choose $f$ and $w$ such that all samples achieve a value of $1$.
	Nevertheless, we cannot conclude anything about states we have not sampled yet: \emph{Without knowledge of the constant, we cannot extrapolate from samples}.
	
	We note the underlying connection to the Nyquist-Shannon sampling theorem~\cite[Thm.~1]{shannon1949communication}.
	Intuitively, the theorem states that, for a function that contains no frequencies higher than $W$, it is completely determined by giving its ordinates at a series of points spaced $0.5 \cdot W$ apart.
	If we know the Lipschitz constant, this gives us a way of bounding the \enquote{frequency} of the value function, and thus allows us to determine it by sampling a finite number of points.
	On the other hand, without the Lipschitz constant, we do not know the frequency and cannot judge whether we are \enquote{undersampling}.
\end{example}
Since we do not assume any particular representation of the transition system, we cannot derive such constants in general.
Instead, these would need to be obtained by, e.g., domain knowledge, or tailored algorithms.
As in previous approaches~\cite{DBLP:conf/uai/GuestrinHK04,DBLP:conf/icml/MeloMR08,DBLP:journals/ejcon/AbateKLP10,DBLP:conf/qest/SoudjaniA11,DBLP:conf/hybrid/TkachevMKA13}, we thus resort to assuming that we are given this constant, offloading this (highly non-trivial) step.
Recall that Lipschitz continuity of the transition function implies Lipschitz continuity of the value function (see \arxivref{\cref{app:assumptions:lipschitz:implication}}{\cite[\appendixref{}B.2.1]{techreport}}), but can potentially be checked more easily.


\subsection{Assumptions Applied: Value Iteration Algorithm} \label{sec:cvi:algorithm}

Before we present our new algorithm, we explain how our assumptions allow us to lift VI to the uncountable domain.
Contrary to the finite state setting, we are unable to store precise values for each state explicitly, since there are uncountably many states.
Hence, the algorithm exploits the Lipschitz-continuity of the value function as follows.
Assume that we know that the value of a state $s$ is bounded from below by a value $l$, i.e.\ $\val(s) \geq l$.
Then, by Lipschitz-continuity of $\val$, we know that the value of a state $s'$ is bounded by $l - \metricstates(s, s') \cdot \lipschitzstates$.
More generally, if we are given a finite set of states $\sampled$ with correct lower bounds $\lowerboundstored \colon \sampled \to [0, 1]$, we can safely extend these values to the whole state space by
\begin{equation*}
	\lowerbound(s) \coloneqq {\max}_{s' \in \sampled} \left( \lowerboundstored(s') - \lipschitzstates \cdot \metricstates(s, s')\right).
\end{equation*}
Since $\val(s) \geq \lowerboundstored(s)$ for all $s \in \sampled$, we have $\val(s) \geq \lowerbound(s)$ for all $s \in \States$, i.e.\ $\lowerbound(\cdot)$ is a valid lower bound.
We thus obtain a lower bound for all of the uncountably many states, described symbolically as a combination of finitely many samples.
See \cref{fig:extrapolation_example} for an illustration.

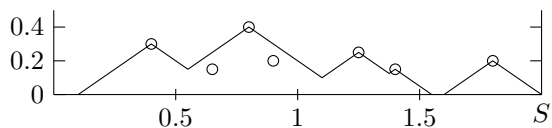
\begin{figure}[t]
	\centering
	\begin{tikzpicture}
		\begin{axis}[xmin=0,xmax=2,ymin=0,ymax=0.5,samples=50,width=8cm,height=2.7cm,
				axis x line=middle,
				axis line style={-},
				enlarge x limits=0,
				enlarge y limits=0,
				xlabel=$\States$,ylabel=\empty,
				xtick={0.5,1,1.5},xticklabels={0.5,1,1.5},
				x label style={anchor=north},
				ytick={0,0.2,0.4},
				every tick/.style={-},
			]
			\addplot[-,no marks,samples=50] plot coordinates {
				(0.1,0)
				(0.4,0.3)
				(0.55,0.15)
				(0.8,0.4)
				(1.1,0.1)
				(1.25,0.25)
				(1.375,0.125)
				(1.4,0.15)
				(1.55,0)
				(1.6,0)
				(1.8,0.2)
				(2,0)
			};
			\addplot[only marks,mark=o] plot coordinates {
				(0.4,0.3)
				(0.8,0.4)
				(1.25,0.25)
				(1.4,0.15)
				(1.8,0.2)

				(0.65,0.15)
				(0.9,0.2)
			};
		\end{axis}
	\end{tikzpicture}
	\caption{Example of the function extension on the set $[0, 2]$ with a Lipschitz constant of $\lipschitzstates = 1$.
	Dots represent stored values in $\lowerboundstored$, while the solid line represents the extrapolated function $\lowerbound$.
	Note that it is possible to have $\lowerboundstored(s) < \lowerbound(s)$, as seen in the graph.
	} \label{fig:extrapolation_example}
\end{figure}

This is sufficient to deal with Markov chains, but for MDPs we additionally need to take care of the (potentially uncountably many) actions.
Recall that value iteration updates state values with the maximum over available actions, $v_{n+1}(s) = \max_{a \in \stateactions(s)} \ExpectedSumMDP{\mdptransitions}{s}{a}{v_n}$.
This is straightforward to compute when there are only finitely many actions, but in the uncountable case obtaining $\lowerbound(s) = \sup_{a \in \stateactions(s)} \lowerbound(s, a)$ is much more involved.
We apply the idea of Lipschitz continuity again, storing values for a set $\sampled$ of state-action pairs instead of only states.
We bound the value of every state-action pair by
\begin{equation} \label{eq:lowerbound_state_action}
	\lowerbound(s, a) \coloneqq {\max}_{(s', a') \in \sampled} \left( \lowerboundstored(s', a') - \metricproduct((s, a), (s', a')) \cdot \lipschitzproduct \right)
\end{equation}
%
Observe that $\lowerbound(s, a)$ is computable and Lipschitz-continuous as well, so by \textbf{Maximum Approximation} we can approximate the bound of any state, i.e.\ $\lowerbound(s) = \max_{a \in \stateactions(s)} \lowerbound(s, a)$, based on such a finite set of values assigned to state-action pairs.
(Recall that $\stateactions(s)$ is compact and $\lowerbound(s, a)$ continuous, hence the maximum is attained.)
Consequently, we can also under-approximate $\ExpectedSumMDP{\mdptransitions}{s}{a}{\lowerbound}$ by \textbf{Transition Approximation}.
To avoid clutter, we omit the following two special cases in the definition of $\lowerbound(s, a)$:
Firstly, if $\sampled = \emptyset$, we naturally set $\lowerbound(s, a) = 0$.
Secondly, if all pairs $(s', a')$ are too far away for a sensible estimate, i.e.\ if \cref{eq:lowerbound_state_action} was yielding $\lowerbound(s, a) < 0$, we also set $\lowerbound(s, a)$ to $0$.

\begin{algorithm*}[t]
	\caption{The Value Iteration (VI) Algorithm for MDPs with general state- and action-spaces.} \label{alg:cvi}
	\begin{algorithmic}[1]
		\Require \textsf{ApproxLower} query with threshold $\xi$, satisfying \textbf{A1--A4}, \textbf{B.VI} and \textbf{C}.
		\Ensure \texttt{yes}, if $\val(\initialstate) > \xi$.
		\State $\sampled \gets \emptyset, \algostep \gets 1$ \Comment{Initialize}
		\While{$\underapprox(\lowerbound(\initialstate), \getprecision) \leq \xi$} 
			\State $(s, a) \gets \getpair$ \Comment{Sample state-action pair}
			\If{$s \in \targetset$} \quad $\lowerboundstored(s, \cdot) \gets 1$ \label{line:alg:cvi:target} \Comment{Handle target states}
			\Else \quad $\lowerboundstored(s, a) \gets \underapprox(\ExpectedSumMDP{\mdptransitions}{s}{a}{\lowerbound}, \getprecision)$ \label{line:alg:cvi:update} \Comment{Update $\lowerboundstored$}
			\EndIf
			\State $\sampled \gets \sampled \union \{(s, a)\}$, $\algostep \gets \algostep + 1$
		\EndWhile
		\State \Return \texttt{yes}
	\end{algorithmic}
\end{algorithm*}

We present VI for MDPs with general state- and action-spaces in \cref{alg:cvi}.
It depends on $\getprecision$, a sequence of precisions converging to zero in the limit, e.g.\ $\getprecision = \frac{1}{\algostep}$.
The algorithm executes the main loop until the current approximation of the lower bound of the initial state $\lowerbound(\initialstate) = \max_{a \in \stateactions(\initialstate)} \lowerbound(\initialstate, a)$ exceeds the given threshold $\xi$.
Inside the loop, the algorithm updates state-action pairs yielded by $\getpair$.
For target states, the lower bound is set to $1$.
Otherwise, we set the bound of the selected pair to an approximation of the expected value of $\lowerbound$ under the corresponding transition.
Here is the crucial difference to VI in the finite setting: Instead of using \cref{eq:vi}, we have to use \cref{eq:lowerbound_state_action} and $\underapprox$, the approximations that exist by assumption, see \cref{sec:vi-assms-basic}.
Since $\getprecision$ converges to zero, the approximations eventually get arbitrarily fine.
The procedure $\getprecision$ may be adapted heuristically in order to speed up computation.
For example, it may be beneficial to only approximate up to $0.01$ precision at first to quickly get a rough overview.
%
%
We show that \cref{alg:cvi} is correct, i.e.\ the stored values (i)~are lower bounds and (ii)~converge to the true values in \arxivref{\cref{app:proofs:lvi}}{\cite[\appendixref{}E.1]{techreport}}.
Here, we only provide a sketch, illustrating the main steps.
%
\begin{theorem} \label{stm:cvi_correct}
	\cref{alg:cvi} is correct under Assumptions \textbf{A1--A4}, \textbf{B.VI}, and \textbf{C}, i.e.\ it outputs \texttt{yes} iff $\val(s) > \xi$.
\end{theorem}

\begin{proof}[Proof sketch]
	First, we show that $\lowerbound_{\algostep}(s) \leq \lowerbound_{\algostep + 1}(s) \leq \val(s)$ by simple induction on the step.
	Initially, we have $\lowerbound_1(s) = 0$, obviously satisfying the condition.
	The updates in Lines~\ref{line:alg:cvi:target} and \ref{line:alg:cvi:update} both keep correctness, i.e.\ $\lowerbound_{\algostep + 1}(s) \leq \val(s)$, proving the claim.
	
	Since $\lowerbound_\algostep$ is monotone as argued above, its limit for $\algostep \to \infty$ is well defined, denoted by $\lowerbound_\infty$.
	By \textbf{State-Action Sampling}, the set of accumulation points of $s_\algostep$ contains all reachable states $\States^{\reach}$.
	We then prove that $\lowerbound_\infty$ satisfies the fixed point equation \cref{eq:value_fixpoint}.
	For this, we use the second part of the assumption on $\getpair$, namely that for every $(s, a) \in \States^{\reach} \times \stateactions$ we get a converging subsequence $(s_{\algostep_k}, a_{\algostep_k})$ where additionally $\mdptransitions(s_{\algostep_k}, a_{\algostep_k})$ converges to $\mdptransitions(s, a)$ in total variation.
	Intuitively, since infinitely many updates occur infinitely close to $(s, a)$, its limit lower bound $\lowerbound_\infty(s, a)$ agrees with the limit of the updates values $\lim_{k \to \infty} \ExpectedSumMDP{\mdptransitions}{s_{\algostep_k}}{a_{\algostep_k}}{\lowerbound_{\algostep_k}}$.
	Since $\lowerbound_\infty$ satisfies the fixed point equation and is less or equal to the value function $\val$, we get the result, since $\val$ is the smallest fixed point.
\end{proof}
\section{Converging Upper Bounds} \label{sec:cbrtdp}

In this section, we present the second set of assumptions, allowing us to additionally compute converging \emph{upper} bounds.
With both lower and upper bounds, we can quantify the progress of the algorithm and, in particular, terminate the computation once the bounds are sufficiently close.
Therefore, instead of only providing a semi-decision procedure for reachability, this algorithm is able to determine the maximal reachability probability up to a given precision.
Thus, we obtain the first algorithm able to handle such general systems with guarantees on its result.
We again present our assumptions together with a discussion of their necessity (\cref{sec:brtdp-assms}), and then introduce the subsequent algorithm and prove its correctness (\cref{sec:brtdp-algo}).
As expected, obtaining this additional information also requires additional assumptions.
On the other hand, quite surprisingly, we can use the additional information of upper bounds to actually \emph{speed up} the computation, as discussed in \cref{sec:discussion-practical}.

As before, our approach is inspired by algorithms for finite MDP, in this case by \emph{Bounded Real-Time Dynamic Programming} (BRTDP) \cite{DBLP:conf/icml/McMahanLG05,DBLP:conf/atva/BrazdilCCFKKPU14}.
BRTDP uses the same update equations as VI, but iterates both lower and upper bounds.
A major contribution of~\cite{DBLP:conf/atva/BrazdilCCFKKPU14} was to solve the long standing open problem of how to deal with \emph{end components}.
These parts of the state space prevent convergence of the upper bounds by introducing additional fixpoints of \cref{eq:value_fixpoint}.
We direct the interested reader to \arxivref{\cref{app:insights:finite_brtdp}}{\cite[\appendixref{}A.2]{techreport}} for further details on BRTDP and insights on the issue of end components.
In the uncountable setting,  these issues arise as well alongside several other, related problems, which we discuss in \cref{sec:brtdp-assm-absorp}.


\subsection{Assumptions}\label{sec:brtdp-assms}

The basic assumptions \textbf{A1}--\textbf{A4} as well as \textbf{Lipschitz continuity} (Assumption \textbf{C}) remain unchanged.
For \textbf{Maximum Approximation} (\textbf{A2}) and \textbf{Transition Approximation} (\textbf{A3}), we additionally require that we are able to \emph{over}-approximate the respective results.
The respective assumptions are denoted by \textbf{A5} and \textbf{A6}, respectively, and both over-approximations by $\overapprox$.
Further, we only require a weakened variant of \textbf{State-Action Sampling}, now called Assumption \textbf{B.BRTDP} instead of Assumption \textbf{B.VI}.
Finally, there is the new Assumption \textbf{D} called \textbf{Absorption}, addressing the aforementioned issue of end components.

\subsubsection{B: Weaker Sampling (Asm.\ B.BRTDP)}

We again assume a $\getpair$ oracle, but, perhaps surprisingly, with \emph{weaker} assumptions.
Instead of requiring it to return \enquote{all} actions, we only require it to yield \enquote{optimal} actions, respective to a given state-action value function.
We first introduce some notation.
Intuitively, we want $\getpair$ to yield actions which are optimal with respect to the \emph{upper bounds} computed by the algorithm.
However, these upper bounds potentially change after each update.
Thus, assume that $f_n \colon \States \times \stateactions \to [0, 1]$ is an arbitrary sequence of computable, Lipschitz continuous, (point-wise) monotone decreasing functions, assigning a value to each state-action pair, and set $\mathcal{F} = (f_1, f_2, \dots)$.
For each state $s \in \States$, set
\begin{equation*}
	\stateactions_\mathcal{F}(s) := \{a \in \stateactions(s) \mid \forall \varepsilon > 0.\ \forall N \in \Naturals.\ \exists n > N.\\ {\max}_{a' \in \stateactions(s)} f_n(s, a') - f_n(s, a) < \varepsilon\},
\end{equation*}
i.e.\ actions that infinitely often achieve values arbitrarily close to the optimum of $f_n$.
Let $\States^{\reach}_\mathcal{F} = \{\last{\finitepath} \mid \finitepath \in \Finitepaths<\MDP, \initialstate> \intersection (\States \times \stateactions_\mathcal{F})^* \times \States\}$ be the set of all states reachable using these optimal actions.\footnote{As in \cref{sec:cvi}, we simplify the definition of $\States^{\reach}_\mathcal{F}$ slightly in order to avoid technical details.}
Essentially, we require that $\getpair$ samples densely in $\States^{\reach}_\mathcal{F} \times \stateactions_\mathcal{F}$.
\begin{description}
	\item[B.BRTDP: State-Action Sampling]
		For any $\varepsilon > 0$, $\mathcal{F}$ as above, $s \in \States^{\reach}_\mathcal{F}$, and $a \in \stateactions_\mathcal{F}(s)$ we have that $\getpair$ a.s.\ eventually yields a pair $(s', a')$ with $\metricproduct((s, a), (s', a')) < \varepsilon$ and $\totalvariation(\mdptransitions(s, a), \mdptransitions(s', a')) < \varepsilon$.
\end{description}
While this new variant may seem much more involved, it is \emph{weaker} than its previous variant, since $\stateactions_\mathcal{F}(s) \subseteq \stateactions(s)$ for each $s \in \States$ and thus also $\States^{\reach}_\mathcal{F} \subseteq \States^{\reach}$.
As such, it also allows for more practical optimizations, which we briefly discuss in \cref{sec:discussion-practical}.



\subsubsection{D: Absorption} \label{sec:brtdp-assm-absorp}

We present our most specific assumption.
While it is \emph{not needed} for correctness, we require it for convergence of the upper bounds to the value and thus for termination of the algorithm.
\begin{description}
	\item[D: Absorption]
	There exists a \emph{known and decidable} set $\sinkset$ (called \emph{sink}) such that $\val(s) = 0$ for all $s \in \sinkset$.
	Moreover, for any $s \in \States$ and strategy $\strategy$ we have $\ProbabilityMDP<\MDP, s><\strategy>[\reach (\targetset \union \sinkset)] = 1$.
\end{description}
Intuitively, the assumption requires that for all strategies, the system will eventually reach a target or a goal state; in other words:
It is not possible to avoid both target and sink infinitely long.
Variants of this assumption are used in numerous settings:
On MDP, it is similar to the contraction assumption, e.g.\ \cite[Chp.~4]{bertsekas1978stochastic}; in stochastic game theory (a two-player extension of MDP) it is called \emph{stopping}, e.g.\ \cite{DBLP:journals/iandc/Condon92}; and, using terms from the theory of the stochastic shortest path problem, we require all strategies to be \emph{proper}, see e.g.\ \cite{DBLP:journals/mor/BertsekasT91}.

This assumption already is important in the finite setting:
There, \textbf{Absorption} is equivalent to the absence of \emph{end components}, which introduce multiple solutions of \cref{eq:value_fixpoint}.
Then, a VI algorithm computing upper bounds can be \enquote{stuck} at a greater fixpoint than the value and thus does not converge~\cite{DBLP:conf/atva/BrazdilCCFKKPU14,DBLP:journals/tcs/HaddadM18}.
\emph{Any} procedure using value iteration thus either needs to exclude such cases or detect and treat them.
Aside from end components, which are the only issue in the finite setting, 
uncountable systems may feature other complex behaviour, such as Zeno-like approaching the target closer and closer without reaching it.

Unfortunately, even just detecting these problems already is difficult.
For the mentioned, restricted setting of probabilistic programs, almost sure termination is $\Pi_0^2$-complete \cite{DBLP:conf/mfcs/KaminskiK15}.
Yet, universal termination with goal set $\targetset \union \sinkset$ is exactly what we require for \textbf{Absorption}.
So, already on a restricted setting (together with a \emph{given} guess for $\sinkset$), we cannot decide whether the assumption holds, let alone treat the underlying problems.
Thus, we decide to exclude this issue and delegate treatment to specialized approaches.

In summary, while this assumption is indeed restrictive, it is the key point that allows us to obtain convergent upper bounds and thus an anytime algorithm.
As argued above, an assumption of this kind seems to be \emph{necessary} to obtain such an algorithm in this generality.



\begin{remark}
	These problems do not occur when considering \emph{finite horizon} or \emph{discounted} properties, which are frequently used in practice.
	For details on treating finite horizon objectives, see \arxivref{\cref{app:step-bounded}}{\cite[\appendixref{}C.1]{techreport}}.
	Discounted reachability with a factor of $\gamma < 1$ is equivalent to normal reachability where at each step the system moves into a sink state with probability $(1 - \gamma)$.
	\textbf{Absorption} is trivially satisfied and our methods are directly applicable.
\end{remark}

\subsection{Assumptions Applied: The Convergent Anytime Algorithm}\label{sec:brtdp-algo}

With our assumptions in place, we are ready to present our adaptation of BRTDP to the uncountable setting.
Compared to VI, we now also store upper bounds, again using Lipschitz-continuity to extrapolate the stored values.
In particular, together with the definitions of \cref{eq:lowerbound_state_action} we additionally set
\begin{equation*}
	\upperbound(s, a) = {\min}_{(s', a') \in \sampled} \left( \upperboundstored(s', a') + \metricproduct((s, a), (s', a')) \cdot \lipschitzproduct \right).
\end{equation*}
We also set $\upperbound(s, a) = 1$ if either $\sampled = \emptyset$ or the above equation would yield $\upperbound(s, a) > 1$.

\begin{algorithm*}[t]
	\caption{The BRTDP algorithm for MDPs with general state- and action-spaces.} \label{alg:cbrtdp}
	\begin{algorithmic}[1]
		\Require \textsf{ApproxBounds} query with precision $\varepsilon$, satisfying \textbf{A1--A6}, \textbf{B.BRTDP}, \textbf{C} and \textbf{D}.
		\Ensure $\varepsilon$-optimal values $(l, u)$.
		\State $\sampled \gets \emptyset, \algostep \gets 1$ \Comment{Initialize}
		\While{$\overapprox\Big(\upperbound(\initialstate), \getprecision\Big) - \underapprox\Big(\lowerbound(\initialstate), \getprecision\Big) \geq \varepsilon$}
			\State $s, a \gets \getpair$ \Comment{Sample stat-action pair}
			\If{$s \in \targetset$} \quad $\lowerboundstored(s, \cdot) \gets 1$ \label{line:alg:cbrtdp:target} \Comment{Handle special cases}
			\ElsIf{$s \in \sinkset$} \quad $\upperboundstored(s, \cdot) \gets 0$ \label{line:alg:cbrtdp:sink}
			\Else \Comment{Update upper and lower bounds}
				\State $\upperboundstored(s, a) \gets \overapprox(\ExpectedSumMDP{\mdptransitions}{s}{a}{\upperbound}, \getprecision)$  \label{line:alg:cbrtdp:update_u}
				\State $\lowerboundstored(s, a) \gets \underapprox(\ExpectedSumMDP{\mdptransitions}{s}{a}{\lowerbound}, \getprecision)$ \label{line:alg:cbrtdp:update_l}
			\EndIf
			\State $\sampled \gets \sampled \union \{(s, a)\}$, $\algostep \gets \algostep + 1$
		\EndWhile
		\State \Return $(\lowerbound(\initialstate), \upperbound(\initialstate))$
	\end{algorithmic}
\end{algorithm*}

We present BRTDP in \cref{alg:cbrtdp}.
It is structurally similar to BRTDP in the finite setting (see \arxivref{\cref{app:insights:finite_brtdp}}{\cite[\appendixref{}A.2]{techreport}}).
The major difference is given by the storage tables $\upperboundstored$ and $\lowerboundstored$ used to compute the current bounds $\upperbound$ and $\lowerbound$, again exploiting Lipschitz continuity.
As before, the central idea is to repeatedly update state-action pairs given $\getpair$. 
If $\getpair$ yields a state of the terminal sets $\targetset$ and $\sinkset$, we update the stored values directly.
Otherwise, we back-propagate the value of the selected pair by computing the expected value under this transition.
Moreover, we again require that $\getprecision$ 
converges to zero.
Note that the algorithm can easily be supplied with a-priori knowledge by initializing the upper and lower bounds to non-trivial values.
Moreover, in contrast to VI, this algorithm is an \emph{anytime} algorithm, i.e.\ it can at any time
provide an approximate solution together with its precision.

Despite the algorithm being structurally similar to the finite variant of \cite{DBLP:conf/atva/BrazdilCCFKKPU14}, the proof of correctness unsurprisingly is more intricate due to the uncountable sets.
We again provide both a simplified proof sketch here and the full technical proof in \arxivref{\cref{app:proofs:cbrtdp}}{\cite[\appendixref{}E.2]{techreport}}.
%
\begin{theorem} \label{stm:cbrtdp_correct}
	\Cref{alg:cbrtdp} is correct under Assumptions \textbf{A1--A6}, \textbf{B.BRTDP}, \textbf{C} and \textbf{D}, and terminates with probability 1.
\end{theorem}

\begin{proof}[Proof sketch]
	We again obtain monotonicity of the bounds, i.e.\ $\lowerbound_\algostep(s, a) \leq \lowerbound_{\algostep + 1}(s, a) \leq \val(s, a) \leq \upperbound_{\algostep + 1}(s, a) \leq \upperbound_\algostep(s, a)$ by induction on $\algostep$, using completely analogous arguments.
	
	By monotonicity, we also obtain well defined limits $\upperbound_\infty$ and $\lowerbound_\infty$.
	Further, we define the difference function $\bounddifference_\algostep(s, a) = \upperbound_\algostep(s, a) - \lowerbound_\algostep(s, a)$ together with its state based counterpart $\bounddifference_\algostep(s)$ and its limit $\bounddifference_\infty(s)$.
	We show that $\bounddifference_\infty(\initialstate) = 0$, proving convergence.
	To this end, similar to the previous proof, we prove that $\bounddifference_\infty$ satisfies a fixed point equation on $\States^{\reach}_+$ (see \textbf{B.BRTDP}), namely $\bounddifference_\infty(s) = \ExpectedSumMDP{\mdptransitions}{s}{a(s)}{\bounddifference_\infty}$ where $a(s)$ is a specially chosen \enquote{optimal} action for each state satisfying $\bounddifference_\infty(s, a(s)) = \bounddifference_\infty(s)$.
	Now, set $\bounddifference_* = \max_{s \in \States^{\reach}_+} \bounddifference_\infty(s)$ the maximal difference on $\States^{\reach}_+$ and let $\States^{\reach}_*$ be the set of witnesses obtaining $\bounddifference_*$.
	Then, $\mdptransitions(s, a(s), \States^{\reach}_*) = 1$: If a part of the transition's probability mass would move to a region with smaller difference, an appropriate update of a pair close to $(s, a(s))$ would reduce its difference.
	Hence, the set of states $\States^{\reach}_*$ is a \enquote{stable} subset of the system when following the actions $a(s)$.
	By \textbf{Absorption}, we eventually have to reach either the target $\targetset$ or the sink $\sinkset$ starting from any state in $\States^{\reach}_*$.
	Since $\bounddifference_\infty(s) = 0$ for all (sampled) states in $\targetset \union \sinkset$ and $\bounddifference_\infty$ satisfies the fixed point equation, we get that $\bounddifference_\infty(s) = 0$ for all states $\States^{\reach}_*$ and consequently $\bounddifference_\infty(\initialstate) = 0$.
\end{proof}

\section{Discussion} \label{sec:discussion}


\subsection{Relation to Algorithms for Finite Systems and Discretization}\label{sec:discussion-related}

Our algorithm directly generalizes the classical value iteration as well as BRTDP for finite MDP by an appropriate choice of $\getpair$.
In value iteration, it proceeds in round-robin fashion, enumerating all state-action pairs.
Note that the algorithm immediately uses the results of previous updates, corresponding to the \emph{Gau{\ss}-Seidel variant} of VI; to exactly obtain synchronous value iteration, we would have to slightly modify the structure for saving the values.
In BRTDP, $\getpair$ simulates paths through the MDP and we update only those states encountered during the simulation.

Approaches based on discretization through, e.g., grids with increasing precision, essentially reduce the uncountable state space to a finite one.
This is also encompassed by $\getpair$, e.g.\ by selecting the grid points in round robin or randomized fashion.
However, our algorithm has the following key advantages when compared to classical discretization. 
Firstly, it avoids the need to grid the whole state space (typically into cells of regular sizes).
Secondly, in discretization, updating the value of one cell does not directly affect the value in other cells; in contrast in our algorithm, knowledge about a state fluently propagates to other areas (by using \cref{eq:lowerbound_state_action}) without being hindered by (arbitrarily chosen) cell boundaries.

\subsection{Extensions}\label{sec:discussion-theory}

We outline possible extensions and augmentations of our approach to showcase its versatility. 

\smallskip\noindent\textbf{Discontinuities}
Our Lipschitz assumption \textbf{C} actually is slightly stronger than required.
We first give an example of a system exhibiting discontinuities and then describe how our approach can be modified to deal with it.
More details are in \arxivref{\cref{app:extensions:discontinuity}}{\cite[\appendixref{}C.2]{techreport}}.
\begin{example}
	Consider a robot navigating a terrain with cliffs, where falling down a cliff immediately makes it impossible to reach the target.
	There, states which are barely on the edge may still reach the goal with significant probability, while a small step to the side results in falling down the cliff and zero probability of reaching the goal.
\end{example}
To solve this example, one could model the cliff as a steep but continuous slope, which would make our approach still possible.
Unfortunately, this might not be very practical, since the Lipschitz constant then is quite large.

However, if we know of discontinuities, e.g.\ the location of cliffs in the terrain the robot navigates, both our algorithms can be extended as follows:
Instead of requiring $\val$ to be continuous on the whole domain, we may assume that we are given a (finite, decidable) partitioning of the state set $\States$ into several sets $\States_i$.
We allow the value function to be discontinuous along the boundaries of $\States_i$ (the cliffs), as long as it remains Lipschitz-continuous inside each $S_i$.
We only need to slightly modify the assumption on $\getpair$ by requiring that for any state-action pair $(s, a)$ with $s \in S_i$ we eventually get a nearby, similarly behaving state-action pair $(s', a')$ \emph{of the same region}, i.e.\ $s' \in \States_i$.
While computing the bounds of a particular state-action pair, e.g.\ $\upperbound(s, a)$, we first determine which partition $S_i$ the state $s$ belongs to and then only consider the stored values of states inside the region $S_i$.

\smallskip\noindent\textbf{Linear Temporal Logic}
In \cite{DBLP:conf/atva/BrazdilCCFKKPU14}, the authors extend BRTDP to LTL queries \cite{DBLP:conf/focs/Pnueli77}.
Several difficulties arise in the uncountable setting.
For example, in order to prove \emph{liveness} conditions, we need to solve the \emph{repeated reachability} problem, i.e.\ whether a particular set of states is reached infinitely often.
This is difficult even for restricted classes of uncountable systems, and impossible in the general case.
In particular, \cite{DBLP:conf/atva/BrazdilCCFKKPU14} relies on analysing end components, which we already identified as an unresolved problem.
We provide further insight in \arxivref{\cref{app:extensions:ltl}}{\cite[\appendixref{}C.3]{techreport}}.
Nevertheless, there is a straightforward extension of our approach to the subclass of \emph{reach-avoid} problems \cite{DBLP:journals/automatica/SummersL10} (or \emph{constrained reachability} \cite{DBLP:journals/iandc/TkachevMKA17})
, see \arxivref{\cref{app:extensions:ltl:reach-avoid}}{\cite[\appendixref{}C.4]{techreport}}.

\subsection{Implementation and Heuristics}\label{sec:discussion-practical}

For completeness, we implemented a prototype of our BRTDP algorithm to demonstrate its effectiveness.
See \arxivref{\cref{app:experiments}}{\cite[\appendixref{}D]{techreport}} for details and an evaluation on both a one- and two-dimensional navigation model.
Our implementation is barely optimized, with no delegation to high-performance libraries. 
Yet, these non-trivial models are solved in reasonable time.
However, since we aim for assumptions that are as general as possible, one cannot expect our generic approach perform on par with highly optimized tools.
Our prototype serves as a proof-of-concept and does not aim to be competitive with specialized approaches. We highlight again that the goal of our paper is \emph{not} to be practically efficient in a particular, restricted setting, but rather to provide general assumptions and theoretical algorithms applicable to all kinds of uncountable systems.

Aside from several possible optimizations concerning the concrete implementation, we suggest two more general directions for heuristics:

\smallskip\noindent\textbf{Adaptive Lipschitz constants}
	As an example, suppose that a robot is navigating mostly flat land close to its home, but more hilly terrain further away.
	The flat land has a smaller Lipschitz constant than the hilly terrain, and thus here we can infer tighter bounds.
	More generally, given a partitioning of the state space and local Lipschitz constants for every subset, we use this local knowledge when computing $\lowerboundstored$ and $\upperboundstored$ instead of using the global Lipschitz constant, which is the maximum of all local ones.
	See \arxivref{\cref{app:extensions:discontinuity}}{\cite[\appendixref{}C.2]{techreport}} for details.

\smallskip\noindent\textbf{$\getpair$-heuristics}
	In \cref{sec:vi-assms-sampling}, we mentioned two simple implementations of $\getpair$.
	Firstly, we can discretize both state and action space, yielding each state-action pair in the discretization for a finite number of iterations, choosing a finer discretization constant, and repeating the process until convergence.
	Assuming that we can sample all state-action pairs in the discretization, this method eventually samples arbitrarily close to \emph{any} state-action pair in $\States \times \stateactions$ and thus trivially satisfies the sampling assumption.
	This intuitively corresponds to executing interval iteration \cite{DBLP:journals/tcs/HaddadM18} on the (increasingly refined) discretized systems.
	Note that this approach completely disregards the reachability probability of certain states and invests the same computational effort for all of them.
	In particular, it invests the same amount of computational effort into regions which are only reached with probability $10^{-100}$ as in regions around the initial state $\initialstate$.
	
	Thus, a second approach is to sample a path through the system at random, following random actions.
	This approach updates states roughly proportional to the probability of being reached, which already in the finite setting yields dramatic speed-ups~\cite{DBLP:conf/concur/KretinskyM19}.
		
	However, we can also use further information provided by the algorithm, namely the upper bounds.
	As mentioned in \cite{DBLP:conf/atva/BrazdilCCFKKPU14}, following \enquote{promising} actions with a large upper bound proves to be beneficial, since actions with small upper bound likely are suboptimal.
	To extend this idea to the general domain, we need to apply a bit of care.
	In particular, it might be difficult to select exactly from the optimal set of actions, since already $\argmax_{a \in \stateactions(s)} \upperbound(s, a)$ might be very difficult to compute.
	Yet, it is sufficient to choose some constant $\xi > 0$ and over-approximate the set of $\xi$-optimal actions in a given state, randomly selecting from this set.
	This over-approximation can easily be performed by, for example, randomly sampling the set of available actions $\stateactions(s)$ until we encounter an action close to the optimum (which can approximate due to our assumptions).
	By generating paths only using these actions, we combine the previous idea of focussing on \enquote{important} states (in terms of reachability) with an additional focus on \enquote{promising} states (in terms of upper bounds).
	This way, the algorithm \emph{learns} from its experiences, using it as a guidance for future explorations.
	
	More generally, we can easily apply more sophisticated learning approaches by interleaving it with one of the above methods.
	For example, by following the learning approach with probability $\nu$ and a \enquote{safe} method with probability $1 - \nu$ we still obtain a safe heuristic, since the assumption only requires limit behaviour.
	As such, we can combine our approach with existing, learning based algorithm by following their suggested heuristic and interleave it with some sampling runs guided by the above ideas.
	In other words, this means that the learning algorithm can focus on finding a reasonable solution quickly, which is then subsequently \emph{verified} by our approach, potentially \emph{improving} the solution in areas where the learner is performing suboptimally.
	On top, the (guaranteed) bounds identified by our algorithm can be used as feedback to the learning algorithm, creating a positive feedback loop, where both components improve each other's behaviour and performance.
%

\section{Conclusion} \label{sec:conc}

In this work, we have presented the first anytime algorithm to tackle the reachability problem for MDP with uncountable state- and action-spaces, giving both correctness and termination guarantees under general assumptions.
The experimental evaluation of our prototype implementation shows both promising results and room for improvements. 

On the theoretical side, we conjecture that \textbf{Assumption D: Absorption} can be weakened if we complement it with an automatic procedure that finds and treats problematic parts of the state space of a certain kind, similar to the collapsing approach on finite MDP \cite{DBLP:journals/tcs/HaddadM18,DBLP:conf/atva/BrazdilCCFKKPU14}.
Note that as the general problem is undecidable, some form of \textbf{Absorption} will remain necessary.
On the practical side, we aim for a more sophisticated tool, applying our theoretical foundation to the full range of MDP, including discrete discontinuities. Moreover, we want to combine the tool with existing ways of identifying the Lipschitz constant.

\clearpage

\bibliography{ref}

\clearpage
\arxivref{
\appendix
\section{Further Insights} \label{app:insights}

\subsection{Details on Finite Value Iteration} \label{app:insights:finite_vi}

Value iteration is a technique used to solve, e.g.\ reachability queries in the finite MDP setting.
It essentially amounts to applying \emph{Bellman iteration} \cite{bellman1966dynamic} (see \cref{eq:vi}) corresponding to the fixed point equation in \cref{eq:value_fixpoint} \cite[Sec.~4.2]{DBLP:conf/sfm/ForejtKNP11}.
It is known that on finite MDP this iteration converges to the true value $\val$ in the limit from below, i.e.\ for all states $s$ we have (i)~$\lim_{n \to \infty} v_n(s) = \val(s)$ and (ii)~$v_n(s) \leq v_{n+1}(s) \leq \val(s)$ for all iterations $n$ \cite[Thm.~7.2.12]{DBLP:books/wi/Puterman94}\footnote{Note that reachability is a special case of \emph{expected total reward}, obtained by assigning a one-time reward of $1$ to each goal state.}.
It is not difficult to construct a system where convergence up to a given precision takes exponential time, but in practice VI often is much faster than methods based on \emph{linear programming} (LP) \cite[Thm.~10.105]{DBLP:books/daglib/0020348}, which in theory has worst-case polynomial runtime and yields precise answers \cite{DBLP:journals/combinatorica/Karmarkar84}.
An important practical issue of VI is the absence of a \emph{stopping criterion}, i.e.\ a straightforward way of determining in general whether the current values $v_n(s)$ are close to the true value function $\val(s)$, as discussed in, e.g.\ \cite[Sec.~4.2]{DBLP:conf/sfm/ForejtKNP11}.

Traditionally, value iteration computes a value for every state of the system at each step.
However, the iteration can also be executed \emph{asynchronously}.
There, the update order may be chosen by heuristics, as long as fairness constraints are satisfied, i.e.\ eventually all states get updated (see e.g.\ \cite{DBLP:books/daglib/0067089}).
Additionally, while VI typically assigns values to each state, one can also store and update values for each state-action pair separately, i.e.\ updating $v_{n+1}(s, a) = \ExpectedSumMDP{\mdptransitions}{s}{a}{v_n}$ and derive $v_n(s)$ by computing $v_n(s) := \max_{a \in \stateactions(s)} v_n(s, a)$.

\subsection{Details on BRTDP for Finite MDP} \label{app:insights:finite_brtdp}

\begin{algorithm}[t]
	\caption{The BRTDP algorithm for finite MDP without end components.} \label{alg:finite_brtdp}
	\begin{algorithmic}[1]
		\Require MDP $\MDP$, state $\initialstate$, precision $\varepsilon$, target set $\targetset$, sink set $\sinkset$, sampler $\getpair$.
		\Ensure $\varepsilon$-optimal values $(l, u)$.
		\State $\sampled \gets \emptyset$, $s_1 \gets \initialstate$, $\upperbound(\cdot, \cdot) \gets 1$, $\lowerbound(\cdot, \cdot) \gets 0$ \Comment{Initialize}
		\While{$\upperbound(\initialstate) - \lowerbound(\initialstate) \geq \varepsilon$}
			\State $s, a \gets \getpair(\upperbound)$ \Comment{Get next state-action pair}
			\If{$s \in \targetset$} \quad $\lowerbound(s, \cdot) \gets 1$ \label{line:alg:brtdp:target} \Comment{Handle special cases}
			\ElsIf{$s \in \sinkset$} \quad $\upperbound(s, \cdot) \gets 0$ \label{line:alg:brtdp:sink}
			\Else \Comment{Update upper and lower bounds}
				\State $\upperbound(s, a) \gets \ExpectedSumMDP{\mdptransitions}{s}{a}{\upperbound}$  \label{line:alg:brtdp:update_u}
				\State $\lowerbound(s, a) \gets \ExpectedSumMDP{\mdptransitions}{s}{a}{\lowerbound}$ \label{line:alg:brtdp:update_l}
			\EndIf
			\State $\sampled \gets \sampled \union \{(s, a)\}$
		\EndWhile
		\State \Return $(\lowerbound(\initialstate), \upperbound(\initialstate))$
	\end{algorithmic}
\end{algorithm}

We briefly summarize the ideas of BRTDP, initially presented in \cite{DBLP:conf/icml/McMahanLG05} and further developed in \cite{DBLP:conf/atva/BrazdilCCFKKPU14}.
We present a formal description of BRTDP, adapted from \cite{DBLP:conf/atva/BrazdilCCFKKPU14}, in \cref{alg:finite_brtdp}.
BRTDP deals with reachability on a \emph{finite} MDP $\MDP = (\States, \Actions, \stateactions, \mdptransitions)$, i.e.\ $\cardinality{\States} < \infty$ and $\cardinality{\Actions} < \infty$, with a given target set $\targetset \subseteq \States$ and precision $\varepsilon > 0$.
The central idea is to apply (asynchronous) value iteration to compute both lower and upper bounds, iterating until the bounds are $\varepsilon$-close to each other.

The MDP is assumed to have no \emph{end components} except in the target set $\targetset$ and a given \emph{sink} set $\sinkset \subseteq \States$.
Intuitively, end components are parts of the state space where the system can remain forever under a particular strategy.
For example, suppose there are two states $s_1, s_2$ where $\mdptransitions(s_1, a_1, s_2) = \mdptransitions(s_2, a_2, s_1) = 1$, i.e.\ the system can go back and forth between $s_1$ and $s_2$ indefinitely.
Thus $(\{s_1, s_2\}, \{a_1, a_2\})$ is an end component.
Technically, such end components introduce additional fixed points to the equation of \cref{eq:value_fixpoint} and applying the value iteration \cref{eq:vi} to upper bounds would not converge to the true value function, see \cite[Ex.~1]{DBLP:conf/atva/BrazdilCCFKKPU14} for more details.
By excluding end components (except targets or sinks), we basically get that $\ProbabilityMDP<\MDP, s><\strategy>[\reach (\targetset \union \sinkset)] = 1$ for any state $s$ and strategy $\strategy$, i.e.\ no matter what we do, with probability 1 we eventually end up in either the sink or the target.
The algorithm further requires that we cannot reach the target once we enter the sink, i.e.\ $\val(s) = 0$ for all states $s \in \sinkset$.
Clearly, the upper bounds of all states in the sink can safely be set to $0$.
Since we excluded any additional fixed point of \cref{eq:value_fixpoint} by our end component assumption, iterations from below and from above converge to the true value.

The algorithm as presented in \cite{DBLP:conf/atva/BrazdilCCFKKPU14} repeatedly samples a path until a target or a sink state is visited and then back-propagates the upper and lower bounds along this path.
In our formulation in \cref{alg:finite_brtdp}, this is captured by $\getpair$ as follows: $\getpair$ first samples a path and then returns the states of the path in reverse order, i.e.\ starting at the target or sink state and ending at the initial state.
Hence, the updates in \cref{line:alg:brtdp:update_u} and \cref{line:alg:brtdp:update_l} are executed on all states on the path, and the information of whether a target or sink was reached is back-propagated.
This is why we allow $\getpair$ to be a \emph{stateful} procedure.

One can also apply this back-propagation globally on all states, which effectively is done in interval iteration \cite{DBLP:journals/tcs/HaddadM18}.
In contrast to that, sampling allows the algorithm to focus on \enquote{important} parts of the system, instead of spending effort on unimportant states.
An interesting observation of \cite{DBLP:conf/atva/BrazdilCCFKKPU14} is that the upper bounds can guide this sampling efficiently.
By choosing actions with a promising upper bound, we always follow the actions which, given our current information, could be the best action, a concept sometimes called \emph{optimism in the face of uncertainty}.
We indicate this in \cref{alg:finite_brtdp} by giving $\getpair$ a parameter $\upperbound$.

\clearpage
\section{Additional Discussion of our Assumptions} \label{app:assumptions}

\subsection{Basic Assumptions} \label{app:assumptions:basic}
We first discuss the standard, usually implicit assumptions.
\paragraph*{Assumption A1: Metric Space}
We require metrics on the state and action spaces in order to define Lipschitz continuity.
More generally, we require a notion of distance to be able to extrapolate from a particular state to its neighbours.
In order to use this notion in the algorithm, the metric naturally has to be computable.
This assumption is given on practically all reasonable systems, in particular when considering a well behaved subset of an Euclidean space, e.g.\ $\States \times \Actions \subseteq \Reals^d$.
Hence, the assumption \textbf{Metric Space} is made implicitly in many works, with $\States$ often assumed to be a subset of $\Reals^d$.
In particular, MDP are often used to model physical processes, which usually are characterized by real valued variables or, more generally, variables which allow for a natural notion of distance.

Note that discrete state and action spaces (or state and action spaces with discrete components) satisfy this assumption, too, using the discrete metric, i.e.\ $\metricstates(x, y) = 0$ if $x = y$ and $1$ otherwise.
Specifically, due to our compactness assumption, the discrete parts of the state space cannot be infinite.
So, in essence, this metric lets us investigate each discrete component separately without extrapolating between them.

\paragraph*{Assumptions A2, A3, A5, and A6: Maximum Approximation \& Transition Approximation}

These assumptions are an immediate consequence of our goal to replicate \cref{eq:vi}.
\textbf{Maximum Approximation} essentially only requires that we can compute / approximate $\stateactions(s)$ and somehow describe this set.
Similarly, \textbf{Transition Approximation} only imposes some minimal knowledge about $\mdptransitions(s, a)$.
Both assumptions can be realised through, for example, dense sampling of $f(a)$ and $g'(s') \coloneqq \mdptransitions(s, a, s') \cdot g(s')$.
See \cref{fig:integration_approximation} for an illustration of the \textbf{Transition Approximation} case.
Note that we only require that $g$ can be approximated up to a certain precision.
However, the overall error introduced by computing these approximations of $g$ can be bounded since we consider probability measures:
In particular, suppose the approximations we obtain are $\hat{g}$, i.e.\ $g(x) - \varepsilon \leq \hat{g}(x) \leq g(x)$.
Then $\ExpectedSumMDP{\mdptransitions}{s}{a}{g} - \varepsilon \leq \ExpectedSumMDP{\mdptransitions}{s}{a}{\hat{g}} \leq \ExpectedSumMDP{\mdptransitions}{s}{a}{g}$.
For \textbf{Maximum Approximation}, the \enquote{default} method is even simpler:
We only need to sample actions from $\stateactions(s)$ with a distance of at most $\varepsilon / L_f$ (where $L_f$ is the Lipschitz constant of $f$).
Then, we can evaluate $f$ at all these positions and know that the true maximum of $f$ is at most $\varepsilon$ larger than the largest sampled value (by Lipschitz continuity of $f$).
The cases of over-approximation are exactly analogous.

\begin{figure}
	\centering
	\begin{tikzpicture}
		\begin{axis}[xmin=-3,xmax=3,ymin=0,samples=100,width=8cm,height=4cm,
			axis x line=middle,
			axis line style={-},
			enlarge x limits=0,
			enlarge y limits=0,
			xlabel=$s'$,ylabel={$\mdptransitions(s, a, s') \cdot g(s')$},
			x label style={anchor=north},
			every tick/.style={-},
			declare function = { f(\x) = 1/(sqrt(2*pi))*exp(-((\x)^2)/(2)) * abs((\x - 0.5) / 2 - floor((\x - 0.5) / 2) - 0.5); }
			]
			\addplot[-,no marks,samples=1000] {f(x)};
			\edef\points{-2,-2.3,-2.5,-1.8,-1.4,-1,-0.5,0.1,0.5,0.75,0.9,1.3,1.6,1.9,2,2.3,2.6,3};
			\pgfmathsetmacro{\slope}{15};
			
			\addplot[only marks,mark=x,samples at=\points] {f(x)};
			
			\pgfplotsinvokeforeach{-2,-2.3,-2.5,-1.8,-1.4,-1,-0.5,0.1,0.5,0.75,0.9,1.3,1.6,1.9,2,2.3,2.6,3}{
				\fill[lightgray] (axis cs:{#1 - f(#1) / tan(\slope)},0) -- (axis cs:#1,{f(#1)}) -- (axis cs:{#1 + f(#1) / tan(\slope)},0);
			}
		\end{axis}
	\end{tikzpicture}
	\caption{
		Illustration to show how we can approximate the successor expectation $\ExpectedSumMDP{\mdptransitions}{s}{a}{g}$ for a Lipschitz continuous $g$ through sampling.
		Here, we chose a Gaussian as successor density as and $g(s)$ a zig-zag function.
		As we know that both functions are Lipschitz continuous, the product is Lipschitz continuous, too.
		Hence, the grey area (which can be computed by evaluating $\mdptransitions(s, a, s') \cdot g(s')$) gives us a safe under-approximation of the total integral.
		By additionally applying the same idea to $\mdptransitions(s, a, s')$ alone, we can safely under-approximate the total probability mass we have already considered.
		Together, we obtain a safe and convergent under-approximation of $\ExpectedSumMDP{\mdptransitions}{s}{a}{g}$.
		Note that this only required Lipschitz continuity of the density function of $\mdptransitions(s, a)$.
		If we know that we are dealing with, e.g., a Gaussian or a uniform distribution, we can heavily optimize this process.
	} \label{fig:integration_approximation}
\end{figure}
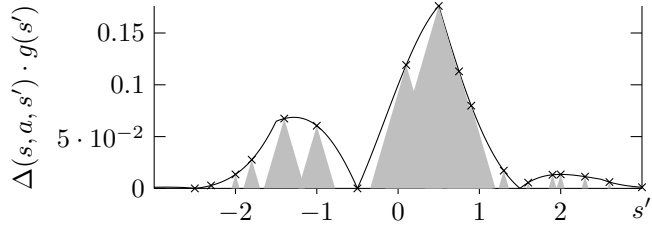

\paragraph*{Assumption A4: Target Computability}

Clearly, we need to be able to decide whether a given state is a target state or not, otherwise the computational problem is not well specified.
We highlight that we do \emph{not} require an explicit description of the target set $\targetset$, we only need a procedure to decide $s \in \targetset$.
As an example, consider the uncountable state space of the real numbers in $[0,1]$. Then, a predicate like $s < \sqrt(2)-1$ allows us to decide for every $s \in [0,1]$ whether it is a target state or not.

%
%

\subsection{Lipschitz Continuity} \label{app:assumptions:lipschitz}

\subsubsection{Lipschitz Continuity of Transition Function implies Lipschitz Continuity of the Value Function}\label{app:assumptions:lipschitz:implication}

We briefly argue that continuity of the transition function (w.r.t.\ total variation) implies continuity of the value function, as similarly shown in, e.g., \cite[Thm.~1]{DBLP:journals/ejcon/AbateKLP10} or \cite[Thm.~3]{DBLP:conf/qest/SoudjaniA11}.
Thus, assume that $\mdptransitions$ is Lipschitz continuous, i.e.\ we have that
\begin{equation*}
	\totalvariation(\mdptransitions(s, a), \mdptransitions(s', a')) \leq L \cdot \metricproduct((s, a), (s', a')).
\end{equation*}
Now, recall that
\begin{equation*}
	\val(s, a) \coloneqq \integral<\hat{s} \in \States>{\val(\hat{s})}{\mdptransitions(s, a)}.
\end{equation*}
Since $0 \leq \val(s) \leq 1$ for all $s \in \States$, we immediately get
\begin{align*}
	\norm{\val(s, a) - \val(s', a')} & \leq \abs*{\integral<\hat{s} \in \States>{\val(\hat{s})}{\mdptransitions(s, a)} - \integral<\hat{s} \in \States>{\val(\hat{s})}{\mdptransitions(s', a')}} \\
		& \leq \totalvariation(\mdptransitions(s, a), \mdptransitions(s', a')) \\
		& \leq L \cdot \metricproduct((s, a), (s', a')).
\end{align*}
To conclude Lipschitz continuity of the value function, one further step is needed.
Recall that $\val(s) \coloneqq \sup_{a \in \stateactions(s)} \val(s, a)$.
If the set of available actions would change abruptly between two nearby states $s$ and $s'$, the continuity of $\val(s, a)$ would not allow us to conclude anything about the continuity of $\val$.\footnote{Observe that this is the underlying reason for our \textbf{State-Action Sampling} assumption, too.}
Thus, we furthermore need a \enquote{continuous} behaviour of the action space.
Formally, we require that for each state-action pair $(s, a)$ and state $s'$, there exists an action $a' \in \stateactions(s')$ such that $\metricproduct((s,a), (s', a')) \leq L' \cdot \metricstates(s, s')$ for some $L' > 0$.\footnote{This assumption is directly implied by the typical, much more restrictive assumption of requiring that $\stateactions$ is constant.}
Then, we can conclude that $\val$ is Lipschitz continuous, too (recall that $\metricstates$ and $\metricproduct$ are compatible).

\subsubsection{Lipschitz Continuous Value Function}

Several works assume Lipschitz continuity of, e.g., the transition function and derive Lipschitz continuity of the value function, while we only assume the latter.
We demonstrate that our assumption is strictly weaker through a small example where the transition function is not even continuous, yet the value function is Lipschitz.

To this end, let $\States = \Actions = [0, 1]$, $\stateactions(s) = \Actions$ and $\mdptransitions(s, a) = \{1 \mapsto 1\}$ if $s < a$, $\{0 \mapsto 1\}$ if $s > a$, and $\{s \mapsto 1\}$ otherwise.
In other words, if $s < a$, we immediately proceed to state $1$, if $s > a$ we go back to $0$, and for $s = a$ we stay on the spot.
As such, a slight change in either state or action may lead to a large change in the transition dynamics.
More concretely, consider the state $s=0.5$ and the action $a=0.5$. We slightly change the state by $\varepsilon$ and get
$\mdptransitions(0.5-\varepsilon,0.5) = \{0 \mapsto 1\}$ or
$\mdptransitions(0.5+\varepsilon,0.5) = \{1 \mapsto 1\}$, which completely differ from each other as well as from the original transition 
$\mdptransitions(0.5,0.5) = \{0.5 \mapsto 1\}$.
Thus, the transition function is not continuous.

Yet, when we choose $\targetset = \{1\}$, the value function is constant (and hence Lipschitz continuous), since all states can trivially reach the target by playing action $1$.

\subsubsection{Formal Definition of the Frequency Markov Chain} \label{app:insights:lipschitz}

In \cref{example:lipschitz}, we formally consider the following MDP (which actually is a Markov chain).
Let $\States = [0, 1] \union \{s_+, s_-\}$ (we add two distinct state for simplicity, however the example can easily be transformed to a completely continuous one) and $\Actions = \{a\}$.
Now, choose some frequency $k \in \Naturals$.
We define the MDP such that the value function is a periodic function with frequency $k$.
Let $f \colon [0, 1] \to [0, 1]$ be a Lipschitz continuous function with $f(0) = f(1)$, e.g., an appropriately scaled sine or similar.
As such, define the transition function such that $\mdptransitions(s, a) = \{s_+ \mapsto f(s \cdot k), s_- \mapsto 1 - f(s \cdot k)\}$ for $0 \leq s \leq \frac{1}{k}$ and $\mdptransitions(s, a) = \{s - \frac{1}{k} \mapsto 1\}$.
Informally, the states between $0$ and $\frac{1}{k}$ obtain a value according to $f$ scaled to this interval and all other states simply move $\frac{1}{k}$ to the left.
Observe that $\val$ is Lipschitz continuous as long as $f$ is Lipschitz continuous.
See \cref{fig:example_lipschitz knowledge} for an illustration of the resulting value function where $f$ is a triangle function and $k = 4$.
Now, observe that for any set of finitely many sampled rational points $\States' \subseteq \States \intersection \mathbb{Q}$, we can choose $k$ and $f$ such that $\val(s) = 1$ for all $s \in \States'$, but there are also uncountably many $s'$ with $\val(s') < \varepsilon$ for every $\varepsilon>0$.
Hence, even though $\val$ is Lipschitz continuous, without knowing the associated constant we cannot conclude anything about neighbouring points.

\clearpage
\section{Extensions and Relaxations} \label{app:extensions}

We outline several possible extensions and augmentations of our approach to showcase the versatility of our assumptions.

\subsection{Finite horizon (step-bounded) reachability} \label{app:step-bounded}

By using the same idea as in \cite{DBLP:conf/atva/BrazdilCCFKKPU14}, our approach is directly able to handle finite horizon reachability, also known as \emph{step-bounded} reachability, i.e.\ the probability of reaching a given target set within $n$ steps.
We simply extend all bound functions with a step counter, e.g.\ $\upperbound(s)$ becomes $\upperbound(s, i)$, denoting an upper bound on the probability of reaching $\targetset$ within $i$ steps.
Similarly, $\getpair$ is supposed to additionally return a step number $i$ between $0$ and $n$.
We then update the lower bound by
\begin{equation*}
	\lowerboundstored(s, a, i) \gets \underapprox(\ExpectedSumMDP{\mdptransitions}{s}{a}{\lowerbound(\cdot, i-1)}, \getprecision)
\end{equation*}
and analogously for the upper bound $\upperboundstored$.
Note that in this case we do not need the sink set $\sinkset$, since after $n$ steps we know that we will not be able to reach $\targetset$ any more, i.e.\ $\val(s, n + 1) = 0$ for all states $s \in \States$.

\subsection{Discontinuities and Local Lipschitz Constants} \label{app:extensions:discontinuity}

\begin{figure}[t]
	\centering
	\begin{tikzpicture}
		\begin{axis}[xmin=0,xmax=2,ymin=0,ymax=1,samples=50,width=9cm,height=3cm,
				axis x line=middle,
				axis line style={-},
				enlarge x limits=0,
				enlarge y limits=0,
				xlabel=$\States$,ylabel=\empty,
				xtick={0.5,1,1.5},xticklabels={0.5,1,1.5},
				x label style={anchor=north},
				ytick={0,0.5,1},
				every tick/.style={-},
			]
			\addplot[-,no marks,samples=50] plot coordinates {
				(0,0.6)
				(0.4,0.4)
				(0.65,0.575)
				(0.8,0.5)
				(1,0.6)
			};
			\addplot[only marks,mark=o] plot coordinates {
				(0.4,0.4)
				(0.8,0.5)
			};
			\addplot[-,no marks,samples=50] plot coordinates {
				(0.0,0.1)
				(0.4,0.3)
				(0.55,0.225)
				(0.8,0.4)
				(1.0,0.3)
			};
			\addplot[only marks,mark=o] plot coordinates {
				(0.4,0.3)
				(0.8,0.4)
			};

			\addplot[-,no marks,dashed] coordinates {(1, 0) (1, 1)};

			\addplot[-,no marks,samples=50] plot coordinates {
				(1,1)
				(1.05,1)
				(1.2,0.7)
				(1.3,0.9)
				(1.45,0.6)
				(1.5,0.7)
			};
			\addplot[only marks,mark=o] plot coordinates {
				(1.05, 1)
				(1.2,0.7)
				(1.45,0.6)
			};
			\addplot[-,no marks,samples=50] plot coordinates {
				(1,0.7)
				(1.05,0.8)
				(1.175,0.55)
				(1.2,0.6)
				(1.425,0.15)
				(1.45,0.2)
				(1.5,0.1)
			};
			\addplot[only marks,mark=o] plot coordinates {
				(1.05,0.8)
				(1.2,0.6)
				(1.45,0.2)
			};

			\addplot[-,no marks,dashed] coordinates {(1.5, 0) (1.5, 1)};

			\addplot[-,no marks,samples=50] plot coordinates {
				(1.5,0.7)
				(1.7,0.5)
				(2,0.8)
			};
			\addplot[only marks,mark=o] plot coordinates {
				(1.7,0.5)
			};
			\addplot[-,no marks,samples=50] plot coordinates {
				(1.5,0.2)
				(1.7,0.4)
				(2,0.1)
			};
			\addplot[only marks,mark=o] plot coordinates {
				(1.7,0.4)
			};
		\end{axis}
	\end{tikzpicture}
	\caption{Example of the function extension on the set $[0, 2]$ when provided with a partitioning $S_1 = [0, 1]$, $S_2 = [1, 1.5]$, $S_3 = [1.5, 2]$ and different Lipschitz constants $\lipschitz_1 = 0.5$, $\lipschitz_2 = 2$, $\lipschitz_3 = 1$.
	As in \cref{fig:extrapolation_example}, the dots represent stored values and solid lines represent the extrapolated functions.
	} \label{fig:discontinuity_example}
\end{figure}
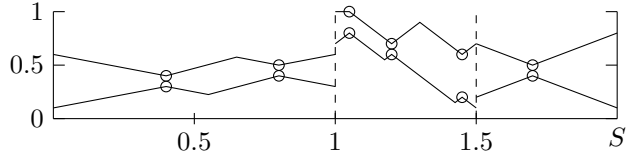

A first, simple extension is given by allowing the value function to be discontinuous at known locations in a well behaved way.
In particular, we assume that we are given a finite partitioning of the state set $\States$ into several sets $\States_i$ for $i \in \{1, \dots, n\}$, i.e.\ $\States = \Union_{i = 1}^n \States_i$ and $\States_i$ pairwise disjoint.
We allow the value function to be discontinuous along the boundaries of $\States_i$, as long as it remains Lipschitz-continuous inside $S_i$.
Note that we do \emph{not} require any special property of $\States_i$ except decidability and Lipschitz continuity of the value function in its interior.
In particular, $\States_i$ does not need to be closed, convex or have a closed form representation.
We only need to slightly modify the assumption on $\getpair$ by requiring that for any state-action pair $(s, a)$ with $s \in S_i$ we eventually get a nearby, similarly behaving state-action pair $(s', a')$ where we also have $s' \in \States_i$.

The only necessary change to the algorithms is the following:
While computing the bounds of a particular state-action pair, e.g.\ $\upperbound(s, a)$, we first determine which partition the state belongs to, i.e.\ find the unique set $S_i$ such that $s \in S_i$, and then only consider states sampled inside this partition.
Then, we define
\begin{equation*}
	\upperbound(s, a) = {\min}_{(s', a') \in \sampled, s' \in S_i} \left( \upperboundstored(s', a') + \lipschitzproduct \cdot \metricproduct((s, a), (s', a')) \right)
\end{equation*}
and $\lowerbound$ analogously, omitting the obvious special cases.
Recall that $\sampled$ is the set of all state-action pairs sampled so far.
It is easy to see that correctness is preserved.
For termination, we can still establish the respective fixed point equations in the same way.

Similarly, we can allow for local Lipschitz constants for each state, i.e.\ an oracle assigns a Lipschitz constant to each state or state-action pair.
Then, whenever we want to compute the bounds of any particular state, we simply use the individual Lipschitz constants for each state-action pair in $\sampled$.
To obtain termination, we only require that the Lipschitz constants are universally bounded.


\subsection{Linear Temporal Logic} \label{app:extensions:ltl}

We discuss how \emph{Linear Temporal Logic} (LTL) \cite{DBLP:conf/focs/Pnueli77} can be approached for uncountable MDP using our approach.
We only briefly define LTL and direct the interested reader to related work, e.g., \cite[Chap.~5]{DBLP:books/daglib/0020348}, \cite[Sec.~7.2]{DBLP:conf/sfm/ForejtKNP11}.
See \cite{DBLP:journals/iandc/TkachevMKA17} for further discussion of LTL on uncountable MDP.
Let $AP$ be a finite, non-empty set of atomic propositions and $a \in AP$ an arbitrary proposition.
Such propositions could, for example, describe \enquote{variable x is larger than 5} or \enquote{the system is in an unsafe state}.
An LTL formula then is given by the following syntax
\begin{equation*}
	\phi ::= a \mid \lnot \phi \mid \phi \land \phi \mid \ltlNext \phi \mid \phi \ltlUntil \phi
\end{equation*}
with the usual shorthand definitions $\lfalse = a \land \lnot a$, $\ltrue = \lnot \lfalse$, and $\phi \lor \psi = \lnot (\phi \land \psi)$.
LTL is evaluated over (infinite) sequences of words, i.e.\ elements of ${(\powerset(AP))}^\omega$.
The logical connectives essentially impose restrictions on the \enquote{current} valuation, i.e.\ the atomic propositions at the beginning of the word.
The $\ltlNext \phi$ operator requires that $\phi$ holds in the next step, while $\phi \ltlUntil \psi$ demands that the formula $\phi$ holds at every step until $\psi$ holds.
Two common derivations are $\ltlFinally \phi := \ltrue \ltlUntil \phi$, requiring that $\phi$ eventually holds in the \textbf{f}uture, and $\ltlGlobally \phi := \lnot \ltlFinally \lnot \phi$ requires that $\phi$ holds \textbf{g}lobally at every position.
As an example, $a \land (X b \lor \ltlFinally \ltlGlobally \lnot a)$ requires that in the first step we have $a$ and either we have $b$ in the next step, or eventually we will never see $a$ again, i.e.\ $a$ only is true finitely often.

In the finite setting, we can equip an MDP with a valuation mapping $\nu \colon \States \to \powerset(AP)$, assigning to each state a set of atomic propositions which hold in this state.
This mapping can directly be lifted to paths, i.e.\ given a path $\infinitepath$ we derive the respective word $\nu(\infinitepath) = \nu(\infinitepath(1)) \nu(\infinitepath(2)) \cdots$.
Thus, we can pose \emph{quantitative LTL queries}, e.g.\ \enquote{what is the maximal probability of satisfying the given formula?}.
It is known that such queries can be reduced to structural pre-computations and then solving a reachability query.
Based on these ideas, \cite{DBLP:conf/atva/BrazdilCCFKKPU14} explains how the BRTDP algorithm can be adapted to accommodate for such an LTL query.

However, several difficulties arise in the uncountable setting.
A central part of the finite-state algorithm is identifying (winning) end components of the \emph{product MDP} (see, e.g., \cite[Sec.~10.3, Sec.~10.6.4]{DBLP:books/daglib/0020348}, for more details on this \emph{automata-theoretic approach} \cite{DBLP:journals/jcss/VardiW86}), or equivalently solving the \emph{repeated reachability} problem.
In particular, to prove \emph{safety} conditions, for example \enquote{remain inside a region forever} ($\ltlGlobally \texttt{in\_region}$), one cannot use sampling alone in general, since even for a single, finite trace it is impossible to give a positive judgement for such an \enquote{infinite} horizon property.
Instead, one needs to analyse the system's transition function to infer knowledge about the infinite horizon behaviour, which is difficult even for restricted classes of uncountable systems, and impossible in our case, since we treat the transition function as a black box.
See \cite[Sec.~4]{DBLP:journals/iandc/TkachevMKA17} for further discussion.

Another problem arises already on very simple, \enquote{smooth} systems with equally simple properties, which we illustrate in the following.
Consider an MDP, where $\States = [-2, 2]$, $\stateactions(s) = [-1, 1]$, and $\mdptransitions(s, a) = \mathrm{unif}([a - 1, a + 1])$, i.e.\ uniformly distributed around the location chosen through $a$.
Furthermore, assume that the goal specified by the LTL formula is to remain in the area $[-1, 1]$ forever.
Note that this query is \enquote{stateless}, it is a simple safety requirement.
It is easy to see that by playing action $0$ from every state we satisfy the goal with probability $1$.
Any other strategy which encounters other actions repeatedly yields an almost sure loss, i.e.\ the probability of satisfying the goal is $0$.
This particularly shows that obtaining the correct action by sampling has probability $0$, even though a sampled path following an unsafe strategy may remain inside the safe area for a very long time.
Moreover, the value function is not continuous, namely $\val(s) = \indicator{[-1, 1]}(s)$, since there is a surely winning strategy for any state in $[-1, 1]$.
Interestingly, the state-action value function $\val(s, a)$ is Lipschitz continuous, namely $\val(s, a) = \frac{1}{2} \min \{2 - a, a + 2\}$ for all $s \in [-1, 1]$.
For example, we have that $\val(0, -1) = 0.5$, since by playing $-1$ we only end up in the \enquote{bad} region with $\frac{1}{2}$ probability, otherwise we can recover by playing optimally.

\subsection{Reach-avoid problems} \label{app:extensions:ltl:reach-avoid}

Despite that it may seem quite difficult to solve this problem in general even on simple systems, we actually can apply our approach to so called reach-avoid problems.
These include, for example, a robot navigating towards a recharge station while avoiding dangerous terrain.
More formally, on top of a reachability query we assume to be given a (measurable and decidable) region to be avoided $U \subseteq \States$.

As in \cref{sec:cbrtdp}, we make our usual assumptions, only that we do not require Lipschitz continuity on the whole state space.
Instead, using the ideas of \cref{app:extensions:discontinuity}, it is sufficient to assume Lipschitz continuity of the value function on $\States \setminus (\targetset \union U)$.
Note that we still require the \enquote{sink}-assumptions.
For simplicity, assume that $\ProbabilityMDP<\MDP, s><\strategy>[\reach (\targetset \union U)] = 1$ for all strategies $\strategy$ and states $s$.
This means that eventually the system either has to reach the target or will fall into an unrecoverable \enquote{error} state, for example running out of energy.
In this case, our methods are directly applicable without any major modifications.
Note that the central idea is that (i)~we can judge whether a path succeeds or fails based on a finite prefix and (ii)~such a success or failure occurs with probability 1 under any strategy.
We conjecture that our approach is applicable to any system-LTL pair which satisfies this criterion.

\clearpage
\section{Evaluation} \label{app:experiments}

We implemented a prototype in Java and evaluated it on two models, which navigate inside one- and two-dimensional state space, respectively.
Our $\getpair$ uses a mixture of global random sampling and path sampling.
To reduce the implementation complexity, our prototype only supports finite action space and assumes that all states have all actions available.
Moreover, currently the implementation only supports uniform or discrete distributions, however it easily can be extended to support further types.
Under- and over- approximation is implemented by a (cached) discretized representation of the lower and upper bound functions.
However, adding uncountable actions is not too different from additional space dimensions.
The experiments were carried out on consumer grade hardware (2.60GHz Intel~i7-9750H CPU, 32 GB RAM).

\begin{figure*}[t]
	\centering
	\includegraphics[width=0.85\textwidth]{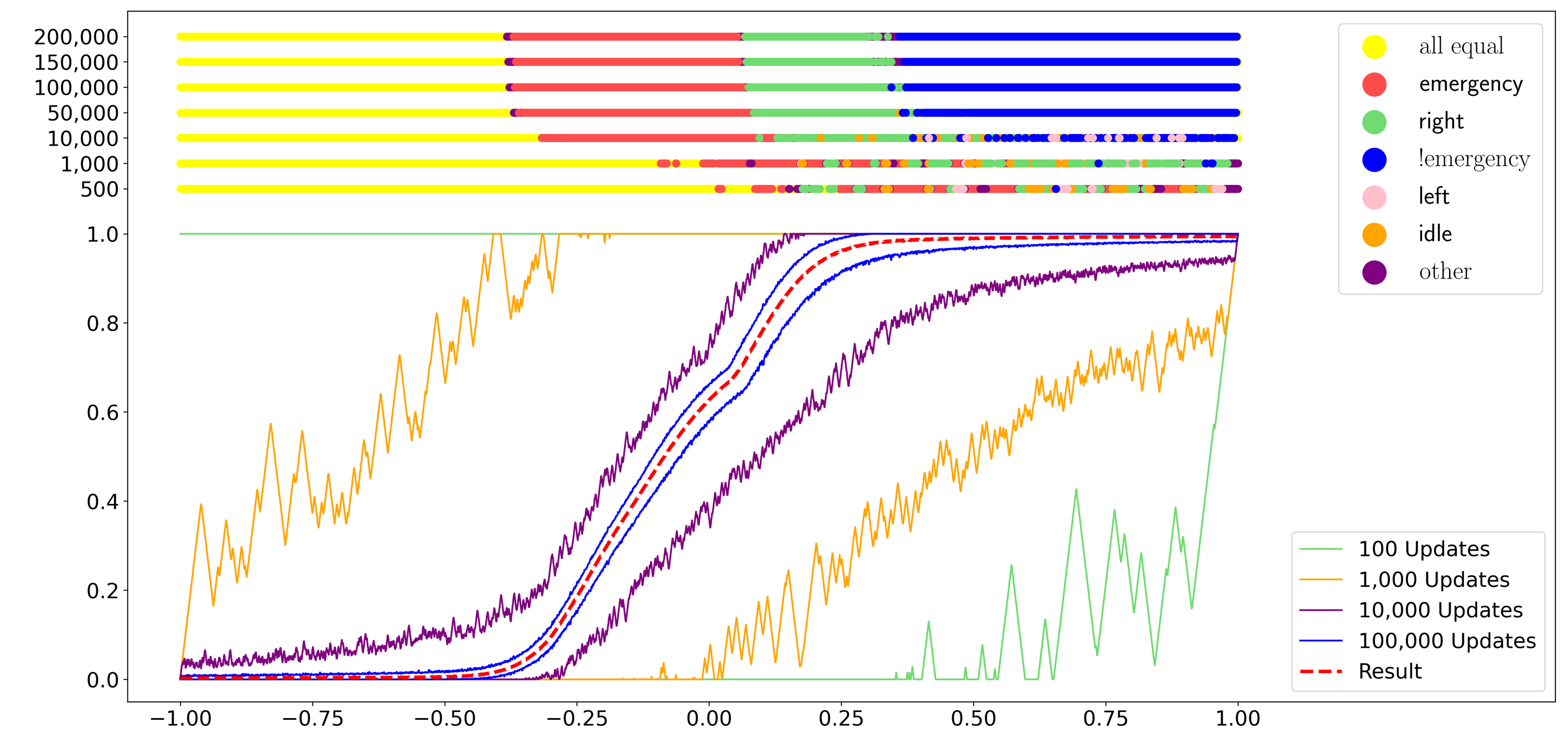}
	\caption{
		Experimental results on the one-dimensional model.
		The bottom part shows the development of the computed upper and lower bounds over the course of the algorithm.
		Furthermore, we show the final result (dashed) as the average of the computed bounds.
		The top part shows the development of the optimal action over time.
		\enquote{all equal} denotes that all actions are optimal, \enquote{!emergency} that all actions except $\textsf{emergency}$ are optimal, and \enquote{other} refers to any other combination of optimal actions not explicitly listed.
	}
	\label{fig:1dimExp}
\end{figure*}

The first model is an MDP over $\States = [-1, 1]$ with $\Actions = \{\textsf{idle}, \textsf{left}, \textsf{right}, \textsf{emergency}\}$.
The target is $\targetset = \{1\}$ and the sink is $\sinkset = \{-1\}$.
Essentially, in each state the system is affected by (i)~its own chosen \enquote{thrust} and (ii)~the gravitational forces of the sinks.
The latter is computed according to textbook physics.
For the former, the three actions $\textsf{idle}$, $\textsf{left}$, and $\textsf{right}$ yield a small force in the respective direction.
The $\textsf{emergency}$ action gives a significantly higher acceleration towards the right side, however also comes with the danger of exploding with a probability of $20\%$, realized by moving to state $-1$ with said probability.
The outcome of each action furthermore is randomized over a continuous interval.
Computation terminates after $236{,}000$ updates, requiring 9 secs and ~0.5~GB RAM to achieve the required precision of $\varepsilon = 0.05$ in the initial state $0$.

\Cref{fig:1dimExp} summarizes the results.
When inspecting the optimal actions, four regions emerge:
States close to $-1$ have practically no chance of escaping.
Thus, all actions have a similar value (yellow), close to zero.
For the states around $0$, the most promising action is the $\textsf{emergency}$ thrust (red) to quickly get outside the gravitational pull of the sink.
When slightly above $0$, just going $\textsf{right}$ (green) without the risk of explosion is safer.
In particular, we identify the trade-off point between $\textsf{emergency}$ and $\textsf{right}$ at approximately $0.05$, where the bounds exhibit a sharp bend.
Lastly, when already close to $1$, all actions except the $\textsf{emergency}$ are optimal (blue), since we still have to avoid the emergency thrust to avoid exploding.
This includes the $\textsf{left}$ action, since the thrust of this action is less than the gravitational pull of the target.

\begin{figure*}[t]
	\centering
	\subfloat[][Value Function]{\includegraphics[width=0.4\textwidth]{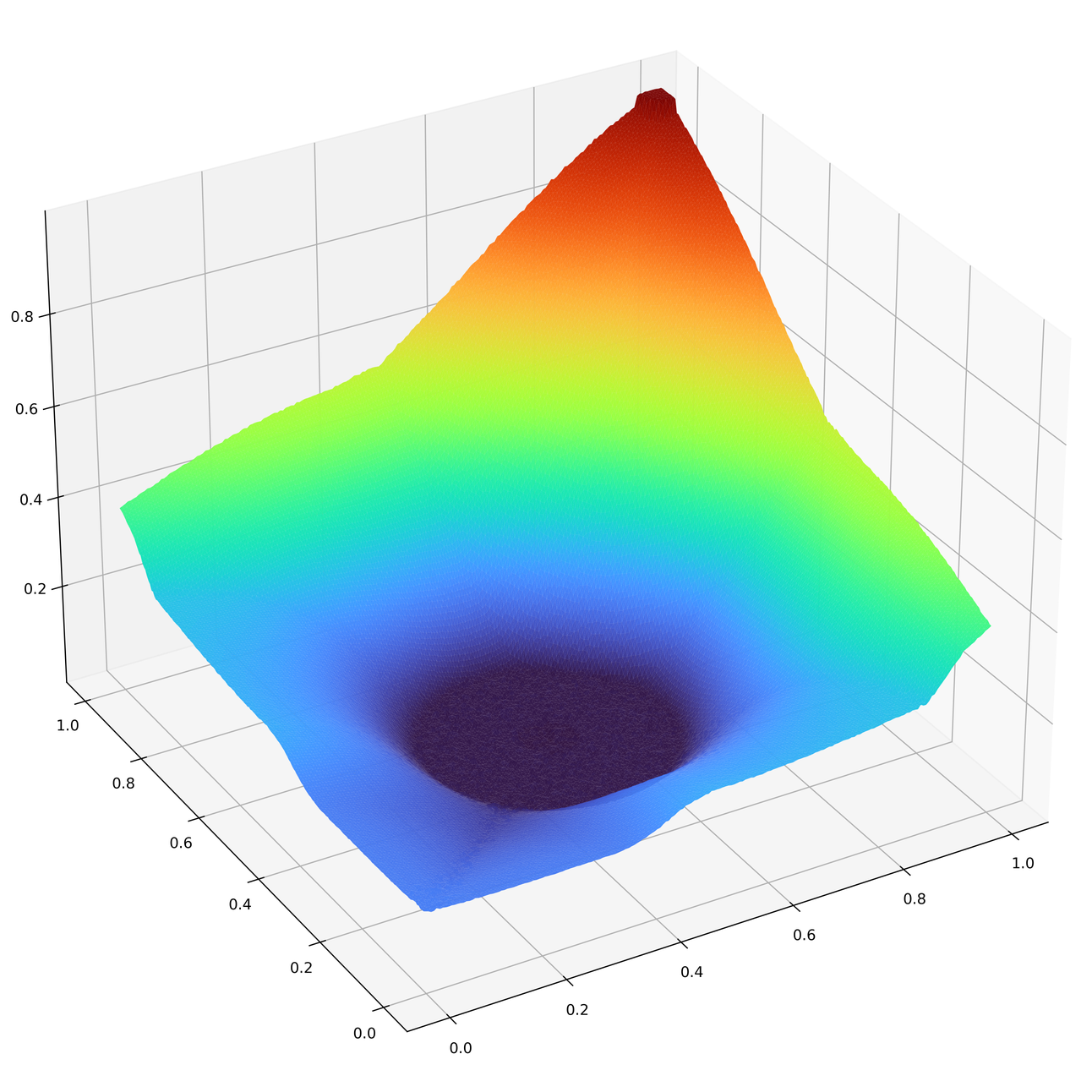}}
	\hfill
	\subfloat[][Optimal actions]{\includegraphics[width=0.47\textwidth]{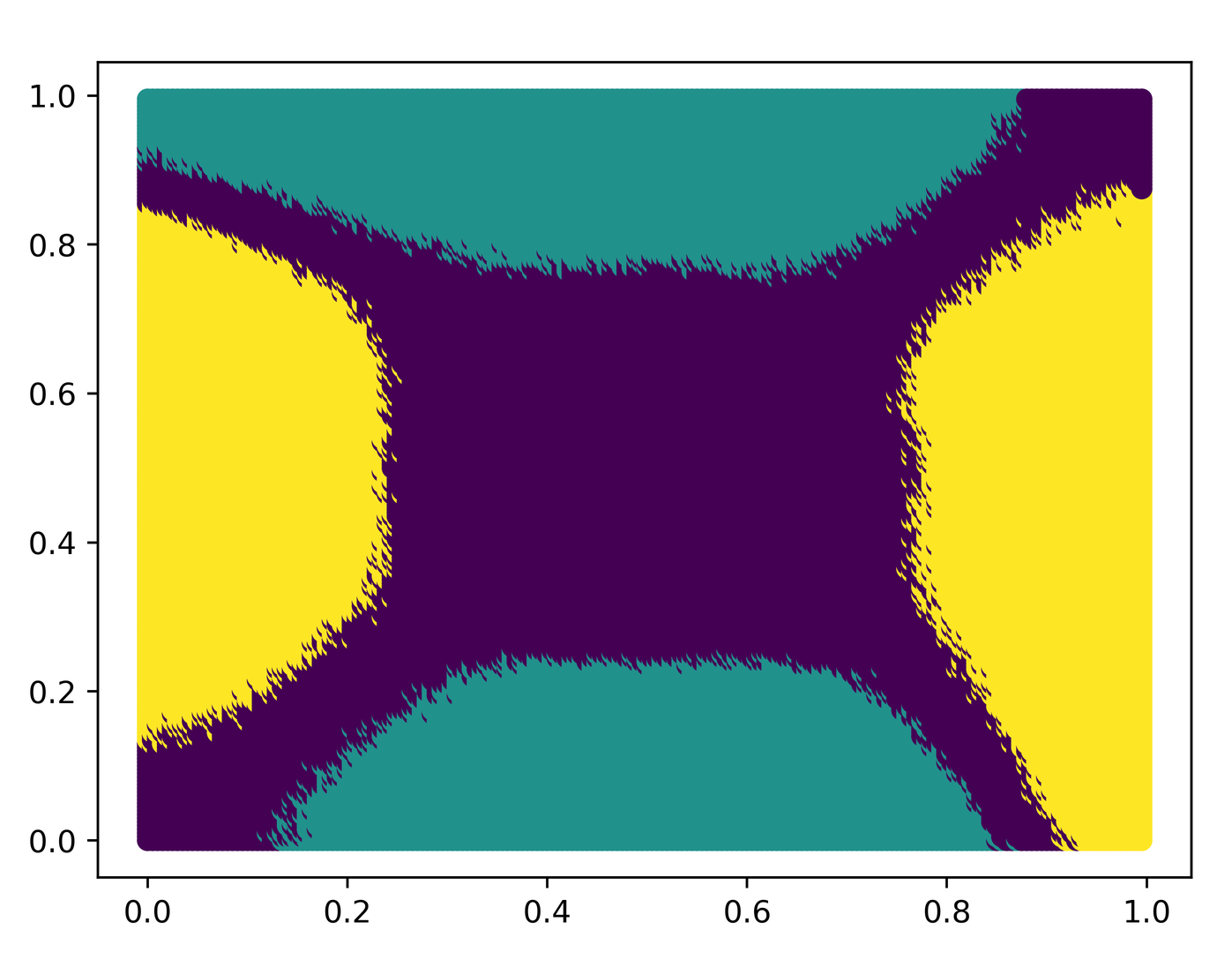}}
	\caption{
		Experimental results on the two-dimensional model.
		On the left, we show the computed value (average of lower and upper bounds).
		On the right, we depict the optimal actions, where yellow corresponds to $\textsf{north}$, green to $\textsf{east}$, and purple means that both actions are roughly equivalent.
	}
	\label{fig:2dimExp}
\end{figure*}

The second model is conceptually similar to the first one:
The system navigates the state space $\States = [0, 1]^2$ using the actions $\Actions = \{\textsf{north}, \textsf{east}\}$ with the expected outcomes.
The target set is $\targetset = \{x \mid \norm{x - (1,1)}_2 < 0.05\}$ with the sink at $\sinkset = \{x \mid \norm{x - (0.5, 0.5)}_2 < 0.05\}$.\footnote{We used regions instead of single points to simplify the implementation of the models.}
Our implementation takes 620 seconds, roughly 3 GB, and $605{,}000$ updates to converge to a precision of $\varepsilon = 0.1$.

\Cref{fig:2dimExp} shows the results for the two-dimensional model.
The output is as expected:
Close to the sink, the value is low, as there is little chance of escaping, and both actions achieve the same values; analogous for the target.
However, in the region around the sink, it is important to choose the right action in order to navigate around the sink as quickly as possible, see \cref{fig:2dimExp}b).

\clearpage
\section{Technical Proofs} \label{app:proofs}

\paragraph*{Definition of Support for General Measures}
As mentioned in the preliminaries, $\support(\mu)$ intuitively describes the \enquote{smallest} set which $\mu$ assigns a value of $1$.
However, consider, for example, the uniform distribution $\lambda$ over the interval $[0, 1]$.
Here, we could always remove single points or even any countable set (which have a measure of $0$ under $\lambda$) and $\lambda$ still assigns a value of $1$ to the remaining set.
Consequently, the support of $\lambda$ would not be well defined, since from any supposedly \enquote{smallest} set, we could again remove a single point and obtain an even smaller set.
For a well-defined notion, let $(X, \mathcal{T})$ be a topological space and $\sigmaalgebra_X = \mathfrak{B}(\mathcal{T})$ the Borel $\sigma$-algebra.
Then, we define $\support(\mu) = \{x \in X \mid x \in N \in \mathcal{T} \Rightarrow \measure(N) > 0\}$, i.e.\ the set of points for which any open neighbourhood $N$ has positive measure.
With this definition, we get $\support(\lambda) = [0, 1]$, as expected.

\paragraph*{Iteration Counter}
In our proofs, we refer to the specific iteration steps of the algorithm and the values of each variable at these steps.
Since every sampled state-action pair is an atomic step of the algorithm, we use the iteration counter $\algostep$ as index of $s$, $a$, $\lowerbound$ and $\upperbound$ to refer to these values.
For example, $s_\algostep$ denotes the state sampled in the $\algostep$-th step.
To ease notation, we assume that any variable which is not changed in step $\algostep$ retains its value.
For example, any lower bound $\lowerboundstored_\algostep(s, a)$ of all state-action action pairs $(s, a)$ which have not been updated in step $\algostep$ still have the same value in step $\algostep + 1$.
We keep this convention throughout all the proofs in \cref{app:proofs}.

\subsection{Detailed Proof of \texorpdfstring{\cref{stm:cvi_correct}}{Theorem 4}} \label{app:proofs:lvi}

We first provide a proof sketch, illustrating the main steps of the proof.
The sketch particularly omits several subtle technicalities.
For example, the precision of values obtained by $\underapprox$ can be non-monotonic and approximation errors might accumulate.
These issues are treated in the main proof.
\begin{proof}[Proof sketch]
	First, we show that $\lowerbound_{\algostep}(s) \leq \lowerbound_{\algostep + 1}(s) \leq \val(s)$ by simple induction on the step.
	Initially, we have $\lowerbound_1(s) = 0$, obviously satisfying the condition.
	The updates in Lines~\ref{line:alg:cvi:target} and \ref{line:alg:cvi:update} both keep correctness, i.e.\ $\lowerbound_{\algostep + 1}(s) \leq \val(s)$, proving the claim.
	
	Since $\lowerbound_\algostep$ is monotone as argued above, its limit for $\algostep \to \infty$ is well defined, denoted by $\lowerbound_\infty$.
	By \textbf{State-Action Sampling}, the set of accumulation points of $s_\algostep$ contains all reachable states $\States^{\reach}$.
	We then prove that $\lowerbound_\infty$ satisfies the fixed point equation \cref{eq:value_fixpoint}.
	For this, we use the second part of the assumption on $\getpair$, namely that for every $(s, a) \in \States^{\reach} \times \stateactions$ we get a converging subsequence $(s_{\algostep_k}, a_{\algostep_k})$ where additionally $\mdptransitions(s_{\algostep_k}, a_{\algostep_k})$ converges to $\mdptransitions(s, a)$ in total variation.
	Intuitively, since infinitely many updates occur infinitely close to $(s, a)$, its limit lower bound $\lowerbound_\infty(s, a)$ agrees with the limit of the updates values $\lim_{k \to \infty} \ExpectedSumMDP{\mdptransitions}{s_{\algostep_k}}{a_{\algostep_k}}{\lowerbound_{\algostep_k}}$.
	Since $\lowerbound_\infty$ satisfies the fixed point equation and is less or equal to the value function $\val$, we get the result, since $\val$ is the smallest fixed point.
\end{proof}

Before we begin with the technical proof of the correctness theorem, we point out a subtle issue in the update computation of the algorithm.
In the proof, we want to show that the lower bounds are monotonically increasing, i.e.\ $\lowerbound_{\algostep}(s, a) \leq \lowerbound_{\algostep + 1}(s, a)$ for any state-action pair $(s, a) \in \States \times \stateactions$.
To obtain this monotonicity in the presence of under-approximations, we need to slightly modify the algorithm in a special case.
In particular, suppose that we sampled some state action pair $(s, a)$ in step $\algostep$ and its lower bound $\lowerboundstored_{\algostep + 1}$ was approximated very precisely.
If we now again sample $(s, a)$ exactly in $\algostep + 1$, the approximation computation may yield a coarser result, consequently we would have $\lowerbound_{\algostep + 2}(s, a) < \lowerbound_{\algostep + 1}$.
To remedy this situation, we modify the algorithm to update $\lowerboundstored$ with the maximum of the current and the computed value.
Note that if we sample a nearby state-action pair $(s', a')$ in the second step, we do not need this special treatment, since then we would simply have that $\lowerboundstored_{\algostep + 2}(s', a') < \lowerbound_{\algostep + 2}(s', a')$, since $\lowerbound_{\algostep + 2}(s', a')$ is then computed based on the stored value of the nearby pair $(s, a)$.
\begin{proof}
	We show that (i)~$\lowerbound_\algostep(s) \leq \val(s)$ for all $s \in \States$ and steps $\algostep$ and (ii)~$\lim_{\algostep \to \infty} \lowerbound_\algostep(\initialstate) = \val(\initialstate)$.

	First, we prove that $\lowerbound_{\algostep}(s) \leq \lowerbound_{\algostep + 1}(s) \leq \val(s)$ for all $s \in \States$ and steps $\algostep$ by induction on $\algostep$.
	Initially, we only need to show that $\lowerbound_1(s, a) \leq \val(s, a)$.
	This clearly is the case since we have $\lowerbound_1(s, a) = 0$, since $\sampled$ is empty and by convention we then set the lower bound to 0.
	Assume we have $\lowerbound_\algostep(s) \leq \val(s)$ for some step $\algostep$.
	The update in Line~\ref{line:alg:cvi:target} is obviously correct, since $\val(s) = 1$ for all $s \in \targetset$.
	Note that in this case we trivially get monotonicity of the updates.
	For the back-propagation in Line~\ref{line:alg:cvi:update}, observe that $\ExpectedSumMDP{\mdptransitions}{s}{a}{\lowerbound_\algostep} \leq \ExpectedSumMDP{\mdptransitions}{s}{a}{\val}$ by induction.
	Moreover, $\val(s, a) = \ExpectedSumMDP{\mdptransitions}{s}{a}{\val}$ by definition.
	Hence $\lowerboundstored_{\algostep + 1}(s_\algostep, a_\algostep) \leq \val(s, a)$.
	Thus, correctness is preserved.
	Monotonicity is directly obtained due to the above discussion.

	Now, we prove that if the algorithm does not terminate we have that $\lim_{\algostep \to \infty} \lowerbound_\algostep(\initialstate) = \val(\initialstate)$.
	Note that when the algorithm does terminate, there is nothing left to prove.
	In the following, we only argue using limit behaviour.
	Hence, because (i) our definition of $\getprecision$ ensures that it converges to 0 in the limit and (ii) by Assumptions \textbf{A2} and \textbf{A3} we can approximate arbitrarily precisely, we have that any arising approximation computation is arbitrarily precise (this in particular means that accumulation of approximation bounds is not an issue).
	We first set up some auxiliary notation.
	Set $\lowerbound_\infty(s, a) = \lim_{\algostep \to \infty} \lowerbound_\algostep(s, a)$, $\lowerbound_\infty(s) = \max_{a \in \stateactions(s)} \lowerbound_\infty(s, a)$ (note that both $\lowerbound_\infty$ are continuous).
	These limits are well-defined due to the above result -- the functions are bounded and monotone.
	Moreover, since the set of state-action pairs is compact, the convergence of $\lowerbound_\algostep$ is uniform.
	Let further $\States_\infty = \{s \mid \forall \varepsilon > 0.\ \forall \algostep.\ \exists \algostep' > \algostep.\ \metricstates(s_{\algostep'}, s) < \varepsilon\}$ the set of all accumulation points of $s_\algostep$, i.e.\ all states to which the algorithm gets arbitrarily close infinitely often.
	By compactness of $\States$, this set is not empty.
	Note that despite the set of all sampled states being countable, $\States_\infty$ may be uncountable.
	Next, for each $s \in \States_\infty$, set $\Actions_\infty(s) = \{a \in \stateactions(s) \mid \forall \varepsilon > 0.\ \forall \algostep.\ \exists \algostep' > \algostep.\ \metricactions(a, a_{\algostep'}) < \varepsilon\}$.
	By our assumption \textbf{State-Action Sampling}, we have that (i)~$\States_\infty$ contains all reachable states, i.e.\ $\States^{\reach} \subseteq \States_\infty$, and (ii)~$\Actions_\infty(s) = \stateactions(s)$ for all $s \in \States_\infty$.
	We now prove that for all $s \in \States_\infty$ we have that either $s \in \targetset$ and $\lowerbound_\infty(s) = 1$ or $\lowerbound_\infty(s, a) = \ExpectedSumMDP{\mdptransitions}{s}{a}{\lowerbound_\infty}$.
	Note that the first case is trivial by the update rule of the algorithm.
	Let $s \in \States^{\reach} \setminus \targetset$ arbitrary and $a \in \stateactions(s)$.
	From our assumption, we obtain a sequence of state-action pairs $(s_{\algostep_k}, a_{\algostep_k})$ which converges to $(s, a)$ and $\mdptransitions(s_{\algostep_k}, a_{\algostep_k})$ converges in total variation to $\mdptransitions(s, a)$.
	By definition of the algorithm, we have that $\lowerbound_{\algostep_k + 1}(s_{\algostep_k}, a_{\algostep_k}) =  \ExpectedSumMDP{\mdptransitions}{s_{\algostep_k}}{a_{\algostep_k}}{\lowerbound_{\algostep_k}}$.
	By uniform convergence and continuity of $\lowerbound_\algostep$, we get $\lim_{k \to \infty} \lowerbound_{\algostep_k + 1}(s_{\algostep_k}, a_{\algostep_k}) = \lim_{k \to \infty} \lowerbound_\infty(s_{\algostep_k}, a_{\algostep_k}) = \lowerbound_\infty(s, a)$.
	By total variation convergence, we obtain that $\lim_{k \to \infty} \ExpectedSumMDP{\mdptransitions}{s_{\algostep_k}}{a_{\algostep_k}}{\lowerbound_\infty} = \ExpectedSumMDP{\mdptransitions}{s}{a}{\lowerbound_\infty}$.
	The desired claim follows.

	To conclude, observe that $\val$ is the least fixed point of the equation system in \cref{eq:value_fixpoint}.
	Since we have shown that $\lowerbound_\infty(s) \leq \val(s)$ for all $s \in \States^{\reach}$ and that $\States^{\reach} \subseteq \States^{\reach}_*$, $\lowerbound_\infty$ equals the least fixed point on these states.
\end{proof}

\clearpage

\subsection{Proof of \texorpdfstring{\cref{stm:cbrtdp_correct}}{Theorem 5}} \label{app:proofs:cbrtdp}

In the algorithm, the $f_n$ of Assumption \textbf{B.BRTDP} corresponds to the computed upper bounds at step $n$ and we write $\States^{\reach}_+$ in the following to denote this set of states generated by the algorithm to simplify notation.
\begin{proof}[Proof sketch]
	We again obtain monotonicity of the bounds, i.e.\ $\lowerbound_\algostep(s, a) \leq \lowerbound_{\algostep + 1}(s, a) \leq \val(s, a) \leq \upperbound_{\algostep + 1}(s, a) \leq \upperbound_\algostep(s, a)$ by induction on $\algostep$, using completely analogous arguments.
	
	By monotonicity, we also obtain well defined limits $\upperbound_\infty$ and $\lowerbound_\infty$.
	Further, we define the difference function $\bounddifference_\algostep(s, a) = \upperbound_\algostep(s, a) - \lowerbound_\algostep(s, a)$ together with its state based counterpart $\bounddifference_\algostep(s)$ and its limit $\bounddifference_\infty(s)$.
	We show that $\bounddifference_\infty(\initialstate) = 0$, proving convergence.
	To this end, similar to the previous proof, we prove that $\bounddifference_\infty$ satisfies a fixed point equation on $\States^{\reach}_+$ (see \textbf{B.BRTDP}), namely $\bounddifference_\infty(s) = \ExpectedSumMDP{\mdptransitions}{s}{a(s)}{\bounddifference_\infty}$ where $a(s)$ is a specially chosen \enquote{optimal} action for each state satisfying $\bounddifference_\infty(s, a(s)) = \bounddifference_\infty(s)$.
	Now, set $\bounddifference_* = \max_{s \in \States^{\reach}_+} \bounddifference_\infty(s)$ the maximal difference on $\States^{\reach}_+$ and let $\States^{\reach}_*$ be the set of witnesses obtaining $\bounddifference_*$.
	Then, $\mdptransitions(s, a(s), \States^{\reach}_*) = 1$: If a part of the transition's probability mass would move to a region with smaller difference, an appropriate update of a pair close to $(s, a(s))$ would reduce its difference.
	Hence, the set of states $\States^{\reach}_*$ is a \enquote{stable} subset of the system when following the actions $a(s)$.
	By \textbf{Absorption}, we eventually have to reach either the target $\targetset$ or the sink $\sinkset$ starting from any state in $\States^{\reach}_*$.
	Since $\bounddifference_\infty(s) = 0$ for all (sampled) states in $\targetset \union \sinkset$ and $\bounddifference_\infty$ satisfies the fixed point equation, we get that $\bounddifference_\infty(s) = 0$ for all states $\States^{\reach}_*$ and consequently $\bounddifference_\infty(\initialstate) = 0$.
\end{proof}

Before the full technical proof, we first prove a small auxiliary lemma, showing that the upper and lower bounds are monotone and sound, i.e.\ $\lowerbound$ is increasing, $\upperbound$ is decreasing, and the value always lies between them.
Note that we again use the same adaptations as discussed in \cref{app:proofs:lvi} in order to ensure monotonicity.
\begin{lemma} \label{stm:bounds_order}
	For any step $\algostep$, state $s$, and action $a \in \stateactions(s)$, we have that $\lowerbound_\algostep(s, a) \leq \lowerbound_{\algostep + 1}(s, a)$, $\upperbound_{\algostep + 1}(s, a) \leq \upperbound_\algostep(s, a)$, and $\lowerbound_\algostep(s, a) \leq \val(s, a) \leq \upperbound_\algostep(s, a)$.
\end{lemma}
\begin{proof}
	We prove by induction on the step $\algostep$.
	Initially, we only need to show that $\lowerbound_1(s, a) \leq \val(s, a) \leq \upperbound_1(s, a)$.
	This clearly is the case since we have $\lowerbound_1(s, a) = 0$ and $\upperbound_1(s, a) = 1$.

	Now, assume we have $\lowerbound_\algostep(s) \leq \val(s) \leq \upperbound_\algostep(s)$ for some step $\algostep$.
	The update in Line~\ref{line:alg:cbrtdp:target} is obviously correct, since $\val(s) = 1$ for all $s \in \targetset$.
	The correctness of Line~\ref{line:alg:cbrtdp:sink} follows directly from our assumption on $\sinkset$ ($\val(s) = 0$ for all $s \in \sinkset$).
	Note that in these two cases we trivially get monotonicity of the updates.
	For the back-propagation in Lines~\ref{line:alg:cbrtdp:update_u} and \ref{line:alg:cbrtdp:update_l}, observe that $\ExpectedSumMDP{\mdptransitions}{s}{a}{\lowerbound_\algostep} \leq \ExpectedSumMDP{\mdptransitions}{s}{a}{\val} \leq \ExpectedSumMDP{\mdptransitions}{s}{a}{\upperbound_\algostep}$ by induction.
	Moreover, $\val(s, a) = \ExpectedSumMDP{\mdptransitions}{s}{a}{\val}$ by definition.
	Hence $\lowerboundstored_{\algostep + 1}(s_\algostep, a_\algostep) \leq \val(s, a) \leq \upperboundstored_{\algostep + 1}(s_\algostep, a_\algostep)$.
	Thus, correctness is preserved.
\end{proof}
With this lemma, we can now prove correctness and termination of our algorithm.
Correctness follows directly from the above lemma.
In order to prove termination, we essentially construct a contradiction based on the \textbf{Absorption} assumption.
\begin{proof}[Proof of \cref{stm:cbrtdp_correct}]
	\textbf{Correctness:}
	Follows directly from \cref{stm:bounds_order}.
	In particular, when the algorithm terminates, we know that $\lowerbound_\algostep(\initialstate) \leq \val(\initialstate) \leq \upperbound_\algostep(\initialstate)$.

	\newcommand{\maxU}{\mathsf{MaxA}}

	\textbf{Termination:}
	We prove by contradiction.
	Thus, assume that the algorithm does not converge, i.e.\ we have $\upperbound_\algostep(\initialstate) - \lowerbound_\algostep(\initialstate) > \varepsilon$ for all steps $\algostep$.
	As before, we only argue using limit behaviour, and again assume that any arising approximation computation is arbitrarily precise.
	First, we need to set up some auxiliary notation.
	Set $\upperbound_\infty(s, a) = \lim_{\algostep \to \infty} \upperbound_\algostep(s, a)$, $\upperbound_\infty(s) = \max_{a \in \stateactions(s)} \upperbound_\infty(s, a)$ (note that both $\upperbound_\infty$ are continuous), and analogously define $\lowerbound_\infty$.
	All limits are well-defined due to \cref{stm:bounds_order} -- the functions are bounded and monotone.
	Moreover, the convergence is uniform due to pointwise convergence on a compact domain.
	Now, define the difference function $\bounddifference_\algostep(s, a) = \upperbound_\algostep(s, a) - \lowerbound_\algostep(s, a)$.
	For any step $\algostep$ and state $s$, let $\maxU_\algostep(s) := \argmax_{a \in \stateactions(s)} \upperbound_\algostep(s, a)$ denote the set of $\upperbound$-optimal actions in state $s$ and set $\bounddifference_\algostep(s) = \max_{a \in \maxU_\algostep(s)} \bounddifference_\algostep(s, a)$ the maximal difference among them.
	It is important to note that $\bounddifference_\algostep(s)$ does not necessarily equal $\max_{a \in \stateactions(s)} \bounddifference_\algostep(s, a)$, but it is easy to show that $\bounddifference_\algostep(s) \geq \upperbound_\algostep(s) - \lowerbound_\algostep(s)$.
	Clearly, $\bounddifference_\infty(s) = \limsup_{\algostep \to \infty} \bounddifference_\algostep(s)$ is well defined, too, since $\bounddifference_\algostep$ is bounded.
	Now, our overall proof strategy is to prove that $\bounddifference_\infty(\initialstate) = 0$ almost surely, since this implies that eventually $\upperbound_\algostep(\initialstate) - \lowerbound_\algostep(\initialstate) \leq \bounddifference_\algostep(s) < \varepsilon$, contradicting our initial assumption.

	Let now $\maxU_\infty(s) = \{a \in \stateactions(s) \mid \forall \varepsilon > 0.\ \forall \algostep.\ \exists \algostep' > \algostep.\ \upperbound_{\algostep'}(s) - \upperbound_{\algostep'}(s, a) < \varepsilon\}$ all actions which infinitely often achieve a value arbitrarily close to the optimum.
	We now show that for any fixed state $s$ there exists an action $a_\infty^{\max}(s) \in \maxU_\infty(s) \subseteq \stateactions(s)$ such that $\bounddifference_\infty(s) = \bounddifference_\infty(s, a)$.
	For every $\algostep$, choose an arbitrary action $a_\algostep^{\max} \in \argmax_{a \in \maxU_\algostep(s)} \bounddifference_\algostep(s, a)$.
	We have that $\upperbound_\algostep(s, a_\algostep^{\max}) = \upperbound_\algostep(s)$ by definition.
	Consequently, the limit $\lim_{\algostep \to \infty} \upperbound_\algostep(s, a_\algostep^{\max})$ is well defined and equals $\upperbound_\infty(s)$.
	Moreover, $\liminf_{\algostep \to \infty} \lowerbound_\algostep(s, a_\algostep^{\max})$ is well defined and we can choose a subsequence of $a_\algostep^{\max}$ obtaining this limit.
	Also, since $\stateactions(s)$ is compact, there exists an accumulation point $a_\infty^{\max}(s) \in \stateactions(s)$ of this subsequence.
	Note that $a_\infty^{\max}(s) \in \maxU_\infty(s)$.
	Together, we get that
	\begin{align*}
		\bounddifference_\infty(s) & = \limsup_{\algostep \to \infty} \left( \upperbound_\algostep(s, a_\algostep^{\max}) - \lowerbound_\algostep(s, a_\algostep^{\max}) \right) \\
			& = \lim_{\algostep \to \infty} \upperbound_\algostep(s, a_\algostep^{\max}) - \liminf_{\algostep \to \infty} \lowerbound_\algostep(s, a_\algostep^{\max}) \\
			& = \upperbound_\infty(s, a_\infty^{\max}(s)) - \lowerbound_\infty(s, a_\infty^{\max}(s)) \\
			& = \bounddifference_\infty(s, a_\infty^{\max}(s)).
	\end{align*}

	Next, we will show that for a particular subset of states, we have that $\bounddifference$ satisfies the fixed point equation, i.e.\ $\bounddifference_\infty(s) = \ExpectedSumMDP{\mdptransitions}{s}{a_\infty^{\max}(s)}{\bounddifference_\infty}$, using our assumptions on $\getpair$.
	Let thus $\States_\infty = \{s \mid \forall \varepsilon > 0.\ \forall \algostep.\ \exists \algostep' > \algostep.\ \metricstates(s, s_{\algostep'}) < \varepsilon\}$ the set of all accumulation points of $s_\algostep$, i.e.\ all states to which the algorithm gets arbitrarily close infinitely often.
	By compactness of $\States$, this set is not empty.
	Note that despite the set of all sampled states being countable, $\States_\infty$ may be uncountable.
	Next, for each $s \in \States_\infty$, set $\Actions_\infty(s) = \{a \in \stateactions(s) \mid \forall \varepsilon > 0.\ \forall \algostep.\ \exists \algostep' > \algostep.\ \metricactions(a, a_{\algostep'}) < \varepsilon\}$.
	Observe that for any state $s \in \States_\infty$, this set is non-empty as well.
	Moreover, due to our assumption $\getpair$, we have with probability 1 that $\States^{\reach}_+ \subseteq \States_\infty$ and $\maxU_\infty(s) \subseteq \Actions_\infty(s)$ for all $s \in \States^{\reach}_+$.

	We now prove that for all $s \in \States^{\reach}_+$ we have that $\bounddifference_\infty(s) = \ExpectedSumMDP{\mdptransitions}{s}{a_\infty^{\max}(s)}{\bounddifference_\infty}$.
	Let $s \in \States^{\reach}_+$ arbitrary.
	From our assumption, we obtain a sequence of state-action pairs $(s_{\algostep_k}, a_{\algostep_k})$ which converges to $(s, a_\infty^{\max}(s))$ and $\mdptransitions(s_{\algostep_k}, a_{\algostep_k})$ converges in total variation to $\mdptransitions(s, a_\infty^{\max}(s))$.
	By definition of the algorithm, we have that $\bounddifference_{\algostep_k + 1}(s_{\algostep_k}, a_{\algostep_k}) =  \ExpectedSumMDP{\mdptransitions}{s_{\algostep_k}}{a_{\algostep_k}}{\bounddifference_{\algostep_k}}$.
	By uniform convergence and continuity of $\bounddifference_\algostep$, we get
	\begin{equation*}
		\lim_{k \to \infty} \bounddifference_{\algostep_k + 1}(s_{\algostep_k}, a_{\algostep_k}) = \lim_{k \to \infty} \bounddifference_\infty(s_{\algostep_k}, a_{\algostep_k}) = \bounddifference_\infty(s, a_\infty^{\max}(s)).
	\end{equation*}
	By total variation convergence, we obtain that
	\begin{equation*}
		\lim_{k \to \infty} \ExpectedSumMDP{\mdptransitions}{s_{\algostep_k}}{a_{\algostep_k}}{\bounddifference_\infty} = \ExpectedSumMDP{\mdptransitions}{s}{a_\infty^{\max}(s)}{\bounddifference_\infty}.
	\end{equation*}
	Together, we obtain the desired claim.

	Now, set $\bounddifference_* = \max_{s \in \States^{\reach}_+} \bounddifference_\infty(s)$ the maximal difference among all such accumulation points -- note that this maximum is obtained because the set $\States^{\reach}_+$ is closed and thus, as a subset of a compact space, also compact.
	Consequently, the set of witnesses $\States^{\reach}_* = \{s \in \States^{\reach}_+ \mid \bounddifference_\infty(s) = \bounddifference_*\}$ is non-empty.
	Recall that we assumed that the algorithm does not converge, in particular we have $\upperbound_\infty(\initialstate) - \lowerbound_\infty(\initialstate) > 0$.
	Hence, $\bounddifference_\infty(\initialstate) > 0$, too, and thus $\bounddifference_* > 0$, since clearly $\initialstate \in \States^{\reach}_+$.
	Next, we show for any state $s \in \States^{\reach}_*$ that $\mdptransitions(s, a_\infty^{\max}(s), \States^{\reach}_*) = 1$.
	Clearly, by definition of $\States^{\reach}_+$ we have $\mdptransitions(s, a_\infty^{\max}(s), \States^{\reach}_+) = 1$.
	However, since $\bounddifference_\infty(s) = \bounddifference_* = \ExpectedSumMDP{\mdptransitions}{s}{a_\infty^{\max}(s)}{\bounddifference_\infty}$, we necessarily have $\mdptransitions(s, a_\infty^{\max}(s), \States^{\reach}_*) = 1$, since for any other successor $s' \in \States^{\reach} \setminus \States^{\reach}_*$, we have $\bounddifference_\infty(s') < \bounddifference_*$.

	Let now $\strategy_*(s) = a_\infty^{\max}(s)$ for all $s \in \States$.
	By our \textbf{Absorption} assumption, we have $\ProbabilityMDP<\MDP, s><\strategy_*>[\reach (\targetset \union \sinkset)] = 1$.
	But, by the above reasoning, we also have that $\ProbabilityMDP<\MDP, s><\strategy_*>[\reach \setcomplement{\States^{\reach}_*}] = 0$ for any $s \in \States^{\reach}_*$.
	Together, this means that $S' = (\targetset \union \sinkset) \intersection \States^{\reach}_*$ satisfies $\ProbabilityMDP<\MDP, s><\strategy_*>[\reach S'] = 1$ for all $s \in \States^{\reach}_*$.
	Now, on the one hand we have that $\bounddifference_\infty(s) = \ExpectedSumMDP{\mdptransitions}{s}{a_\infty^{\max}(s)}{\bounddifference_\infty}$ for all $s \in \States^{\reach}_*$.
	On the other hand, we have that $\bounddifference_\infty(s) = 0$ for all $s \in \States_\infty \intersection (\targetset \union \sinkset)$, in particular for all $s \in S'$.
	Together, we get that $\bounddifference_\infty(s) = 0$ for all $s \in \States^{\reach}_*$, yielding the contradiction.
\end{proof}

}
{}

\end{document}